\documentclass{scholarly-arxiv}

\usepackage{cmll}
\DeclareSymbolFont{stmry}{U}{stmry}{m}{n}
\DeclareMathDelimiter\llbracket{\mathopen}{stmry}{"4A}{stmry}{"71}
\DeclareMathDelimiter\rrbracket{\mathclose}{stmry}{"4B}{stmry}{"79}
\DeclareMathSymbol\llparenthesis\mathopen{stmry}{"4C}
\DeclareMathSymbol\rrparenthesis\mathclose{stmry}{"4D}
\DeclareMathSymbol\subsetplus\mathrel{stmry}{"44}
\DeclareMathSymbol\supsetplus\mathrel{stmry}{"45}
\DeclareMathSymbol\inplus\mathrel{stmry}{"41}
\DeclareMathSymbol\fatsemi\mathbin{stmry}{"23}
\DeclareMathSymbol\bigsqcap\mathop{stmry}{"64}
\DeclareMathSymbol\fatslash\mathbin{stmry}{"28}
\DeclareMathSymbol\fatbslash\mathbin{stmry}{"29}

\DeclareMathAlphabet{\mathpzc}{OT1}{pzc}{m}{it}

\newcommand\Mid{\mathrel{|}}
\newcommand\Coloneqq{\mathrel{\mathop{::}}=}
\newcommand\paren[1]{\left(#1\right)}
\renewcommand\brack[1]{\left[#1\right]}
\newcommand\set[1]{\left\{#1\right\}}
\newcommand\setcompr[2]{\left\{#1~\middle|~#2\right\}}
\newcommand\tuple[1]{\left\langle#1\right\rangle}

\newcommand\eqdef{\mathrel{\mathop{:}}=}
\renewcommand\P{\mathcal{P}}
\newcommand\pset[1]{\P{\paren{#1}}}
\newcommand\fpset[1]{\mathcal{P}_f\paren{#1}}
\newcommand\Nat{\mathbb{N}}

\renewcommand\epsilon\varepsilon
\newcommand\sem[1]{\left\llbracket {#1}\right\rrbracket}

\newcommand\subsetsim{%
  \mathrel{%
    \ooalign{%
      \raise0.2ex\hbox{$\subset$}\cr\hidewidth\raise-0.8ex\hbox{\scalebox{0.9}{$\sim$}}%
      \hidewidth\cr%
    }%
  }%
}

\makeatletter
\newcommand{\superimpose}[2]{%
  {\ooalign{$#1\@firstoftwo#2$\cr\hfil$#1\@secondoftwo#2$\hfil\cr}}}
\makeatother

\DeclareDocumentCommand\precong{ d<> O{H} D(){A} }{%
  \IfValueTF{#1}
  {\mathrel{\leq_{#2_{#3}}^{#1}}}
  {\mathrel{\leq_{#2_{#3}}}}%
}

\DeclareDocumentCommand\cl{ O{\Theta} D(){A} }{%
  {\downarrow^{#1_{#2}}}
}

\NewDocumentCommand\tree{ s d|| d() O{A} }{%
  \mathbbm{T}%
  \IfValueT{#3}{%
    ^{#3}%
  }%
  \IfValueT{#2}{%
    _{#2}%
  }%
  \IfBooleanF{#1}{%
    {\paren{#4}}%
  }%
}

\NewDocumentCommand\context{ D<>{\tree*} d|| d() O{A} }{%
  {#1}%
  \IfValueT{#3}{%
    ^{#3}%
  }%
  \IfValueT{#2}{%
    _{#2}%
  }%
  {\brack{#4}}%
}

\newcommand\myterms{\mathcal{T}}
\NewDocumentCommand\terms{ s D<>{\Sigma} O{A} }{%
  \myterms^{#2}%
  \IfBooleanF{#1}{%
    {\paren{#3}}%
  }%
}

\NewDocumentCommand\contexts{ D<>{\Sigma} d() O{A} }{%
  \IfValueTF{#2}{%
    \context<\myterms>(#2)[#3]
  }{%
    \context<\myterms>[#3]%
  }%
}

\renewcommand\vec[1]{\overrightarrow{#1}}

\makeatletter
\newcommand\XXreldoublebar{\mathrel{\smash=}}
\newcommand{\XXRightarrowfill@}[1]{%
\m@th \setboxz@h {$#1\XXreldoublebar $}\ht \z@ \z@ 
$#1\copy \z@ 
\mkern -6mu\cleaders \hbox {$#1\mkern -2mu\box \z@ \mkern -2mu$}\hfill 
\mkern -6mu
\mathord \Rightarrow $}
\newcommand{\Overrightarrow}{\mathpalette{\overarrow@\XXRightarrowfill@}}
\makeatother

\DeclareFontFamily{U} {MnSymbolA}{}

\DeclareFontShape{U}{MnSymbolA}{m}{n}{
  <-6> MnSymbolA5
  <6-7> MnSymbolA6
  <7-8> MnSymbolA7
  <8-9> MnSymbolA8
  <9-10> MnSymbolA9
  <10-12> MnSymbolA10
  <12-> MnSymbolA12}{}
\DeclareFontShape{U}{MnSymbolA}{b}{n}{
  <-6> MnSymbolA-Bold5
  <6-7> MnSymbolA-Bold6
  <7-8> MnSymbolA-Bold7
  <8-9> MnSymbolA-Bold8
  <9-10> MnSymbolA-Bold9
  <10-12> MnSymbolA-Bold10
  <12-> MnSymbolA-Bold12}{}
\DeclareSymbolFont{MnSyA} {U} {MnSymbolA}{m}{n}
\DeclareMathSymbol\downpitchfork{\mathrel}{MnSyA}{139}
\newcommand\trunkfun\downpitchfork

\NewDocumentCommand\syntax{D<>{\Alg} d()}{
  \IfNoValueTF{#2}{\mathcal{E}^{#1}}
  {{\mathcal{E}^{#1}}{\paren{#2}}}
}
\NewDocumentCommand\Algleq{s D<>{\Alg} O{A}}{
  \IfBooleanTF{#1}
  {\leqq_{#2,#3}}
  {\leqq_{#2}}}
\NewDocumentCommand\Algeq{s D<>{\Alg} O{A}}{
  \IfBooleanTF{#1}
  {\equiv_{#2,#3}}
  {\equiv_{#2}}}
\NewDocumentCommand\interpr{s D<>{\Alg} O{A} d()}{
  \IfBooleanTF{#1}
  {\mathcal{I}^{#2}_{#3}}
  {\mathcal{I}^{#2}}%
  \IfValueT{#4}{\paren{#4}}
}

\NewDocumentCommand\Alg{ d<> }{%
  \IfValueTF{#1}{\mathfrak{A}_{#1}}{\mathfrak{A}}%
}

\NewDocumentCommand\Aleqq{d<> o}{
  \IfValueTF{#2}{
    \leqq_{#2}
  }{
    \IfValueTF{#1}
    {\leqq_{\Alg<#1>}}
    {\leqq_{\Alg}}
  }
}
\NewDocumentCommand\Aeqq{d<> o}{
  \IfValueTF{#2}{
    \equiv_{#2}
  }{
    \IfValueTF{#1}
    {\equiv_{\Alg<#1>}}
    {\equiv_{\Alg}}
  }
}

\NewDocumentCommand\expr{ d|| d<> o d()}
{%
  \IfValueTF{#1}{%
    \IfValueTF{#4}{%
      {\mathbbm{E}_{#1}}{\paren{#4}}%
    }{%
      \IfValueTF{#2}{%
        {\mathbbm{E}_{#1}}{{#2}}%
      }{%
        \IfValueTF{#3}{%
          {\mathbbm{E}_{#1}}{\brack{#3}}%
        }{%
          \mathbbm{E}_{#1}%
        }%
      }%
    }%
  }{%
    \IfValueTF{#4}{%
      \IfValueTF{#2}{%
        {\mathbbm{E}}{{#2}}{\paren{#4}}%
      }{%
        \IfValueTF{#3}{%
          {\mathbbm{E}}{\brack{#3}}{\paren{#4}}%
        }{%
          \mathbbm{E}{\paren{#4}}%
        }%
      }%
    }{%
      \IfValueTF{#2}{%
        {\mathbbm{E}}{{#2}}%
      }{%
        \IfValueTF{#3}{%
          {\mathbbm{E}}{\brack{#3}}%
        }{%
          \mathbbm{E}%
        }%
      }%
    }%
  }%
}

\newcommand\Traces{\mathcal{T}}

\NewDocumentCommand\traces{ D||{\Traces} d:: d<> o d()}
{%
  \IfValueTF{#2}{%
    \IfValueTF{#5}{%
      \IfValueTF{#3}{%
        {\mathbbm{T}_{#1}^{#2}}{{#3}}\paren{#5}
      }{%
        \IfValueTF{#4}{%
          {\mathbbm{T}_{#1}^{#2}}{\brack{#4}}\paren{#5}
        }{%
          {\mathbbm{T}_{#1}^{#2}}{\paren{#5}}
        }%
      }%
    }{%
      \IfValueTF{#3}{%
        {\mathbbm{T}_{#1}^{#2}}{{#3}}%
      }{%
        \IfValueTF{#4}{%
          {\mathbbm{T}_{#1}^{#2}}{\brack{#4}}%
        }{%
          \mathbbm{T}_{#1}^{#2}%
        }%
      }%
    }%
  }{%
    \IfValueTF{#5}{%
      \IfValueTF{#3}{%
        {\mathbbm{T}_{#1}}{{#3}}\paren{#5}
      }{%
        \IfValueTF{#4}{%
          {\mathbbm{T}_{#1}}{\brack{#4}}\paren{#5}
        }{%
          {\mathbbm{T}_{#1}}{\paren{#5}}
        }%
      }%
    }{%
      \IfValueTF{#3}{%
        {\mathbbm{T}_{#1}}{{#3}}%
      }{%
        \IfValueTF{#4}{%
          {\mathbbm{T}_{#1}}{\brack{#4}}%
        }{%
          \mathbbm{T}_{#1}%
        }%
      }%
    }%
  }%
}
\NewDocumentCommand\lang{ D||{\Traces} d<> o d()}
{%
  \IfValueTF{#4}{%
    \IfValueTF{#2}{%
      {\mathbbm{L}_{#1}}{{#2}}{\paren{#4}}%
    }{%
      \IfValueTF{#3}{%
        {\mathbbm{L}_{#1}}{\brack{#3}}{\paren{#4}}%
      }{%
        \mathbbm{L}_{#1}{\paren{#4}}%
      }%
    }%
  }{%
    \IfValueTF{#2}{%
      {\mathbbm{L}_{#1}}{{#2}}%
    }{%
      \IfValueTF{#3}{%
        {\mathbbm{L}_{#1}}{\brack{#3}}%
      }{%
        \mathbbm{L}_{#1}%
      }%
    }%
  }%
}

\newcommand\clT{{\downarrow^{\Traces}}}

\NewDocumentCommand\I{ D||{\mathcal{I}} d()}{
  \IfValueTF{#2}{%
    \sem{#2}_{#1}
  }{%
    {#1}
  }%
}

\NewDocumentCommand\Ic{ D||{\mathcal{I}} d()}{
  \IfValueTF{#2}{%
    \sem{#2}_{#1}^{\Traces}
  }{%
    {\clT\circ {#1}}
  }%
}

\NewDocumentCommand\AExpr{ s D<>{\Alg} D(){A} o}
{\mathcal{E}_{#2}{\IfBooleanF{#1}{\IfValueTF{#4}{\brack{#4}}{\paren{#3}}}}}
\NewDocumentCommand\Terms{ s D<>{\Alg} D::{} D(){A} o}
{\mathcal{T}^{#3}_{#2}{\IfBooleanF{#1}{\IfValueTF{#5}{\brack{#5}}{\paren{#4}}}}}
\NewDocumentCommand\ALang{ s D<>{\Alg} D(){A}}
{\mathcal{L}_{#2}{\IfBooleanF{#1}{\paren{#3}}}}
\NewDocumentCommand\AI{s D<>{\Alg} o d()}{
  \IfValueTF{#4}{%
    \IfValueTF{#3}{%
      \IfBooleanTF{#1}{%
        \sem{#4}^{#3}
      }{%
        \sem{#4}^{#3}_{#2}
      }%
    }{%
       \IfBooleanTF{#1}{%
        \sem{#4}
      }{%
        \sem{#4}_{#2}
      }%
    }%
  }{%
    \IfValueTF{#3}{%
      \IfBooleanTF{#1}{%
        \mathcal{I}^{#3}
      }{%
        \mathcal{I}^{#3}_{#2}        
      }%
    }{%
       \IfBooleanTF{#1}{%
         \mathcal{I}
      }{%
        \mathcal{I} _{#2}       
      }%
    }%
  }%
}

\newcommand\X{\mathfrak{X}}
\NewDocumentCommand\latsem{ O{\X} d() }{
  \IfValueTF{#2}{%
    \sem{#2}_{#1}%
  }{%
    \mathfrak{i}_{#1}%
  }%
}

\NewDocumentCommand\Coh{ d() }{\mathscr C\IfValueT{#1}{{\paren{#1}}}}
\NewDocumentCommand\Cohf{ d() }{\mathscr C_f\IfValueT{#1}{{\paren{#1}}}}
\let\oldcoh\coh
\RenewDocumentCommand\coh{ D(){} }{\mathrel{\oldcoh_{#1}}}
\let\oldincoh\sincoh
\RenewDocumentCommand\incoh{ D(){} }{\mathrel{\oldincoh_{#1}}}
\NewDocumentCommand\atobs{ d() }{O\IfValueT{#1}{{\paren{#1}}}}

\NewDocumentCommand\contains{ D(){} }{\mathrel{\preceq_{#1}}}
\NewDocumentCommand\Obs{ d() }{\mathbb{O}\IfValueT{#1}{{\paren{#1}}}}
\NewDocumentCommand\Obsf{ d() }{\mathbb{O}_f\IfValueT{#1}{{\paren{#1}}}}
\NewDocumentCommand\closure{ D(){} }{\mathop{\downarrow^{#1}}}
\NewDocumentCommand\closuref{ D(){} }{\mathop{\downarrow^{#1}_f}}

\newcommand\Tlat{\mathcal T_{\text{lat}}}
\newcommand\Tobs{\mathcal T_{\text{obs}}}

\begin{document}
\title{Observation algebras:\\ Heyting algebra over coherence spaces}
\titlerunning{Observation algebra}
\author{Paul Brunet}
\authorrunning{P. Brunet}

\institute{%
  Universit\'e Paris-Est Cr\'eteil\\
  \ttfamily\href{https://paul.brunet-zamansky.fr}{paul.brunet-zamansky.fr}\\
  \email{paul@brunet-zamansky.fr}
}

\maketitle              % typeset the header of the contribution
\begin{abstract}
  In this report, we introduce observation algebras, constructed by considering the downclosed subsets of a coherence space ordered by reverse inclusion.
  These may be interpreted as specifications of sets of events via some predicates with some extra structure.
  We provide syntax for these algebras, as well as axiomatisations.
  We establish completeness of these axiomatisations in two cases: when the syntax is that of bounded distributive lattices (conjunction, disjunction, top, and bottom), and when the syntax also includes an implication operator (in the sense of Heyting algebra), but the underlying coherence space satisfies some tractability condition.
  We also provide a product construction to combine graphs and their axiomatisations, yielding a sound and complete composite system. This development has been fully formalised in Rocq.
\end{abstract}

% 
% 

% \section*{Introduction}
% \phantomsection
% \addcontentsline{toc}{section}{Introduction}
% \label{sec:intro}
\section{Introduction}
\label{sec:intro}
When designing formal semantics for program verification, one needs a language (logic) to describe states: such a language forms the static side of an otherwise dynamic framework.
In many cases, boolean algebras (BA), i.e. classical propositional logic, are used for this purpose.
Examples include the KAT family of models~\cite{kozen_kleene_1997}, as well as  (concurrent) Kleene algebras with boolean observations ~\cite{kappe_kleene_2019,kbswz20}.
However, as discussed in~\cite{wbdkrs20}, this is not the most useful model in the presence of concurrent behaviour.
Indeed, in \textit{op. cit.} it was argued that a more general class of logics, namely \emph{pseudo-complemented distributive lattices} or PCDL, were more appropriate.
A tailor-made instance of this logic was presented then, and used to model weak memory properties.
%
% The critical difference between BA and PCDL may be discussed in terms of scope.
% In boolean logic, a sentence is interpreted as a set of valuations, assigning a truth value to each atomic proposition.
% This means in some sense that each formula ranges over the whole model.
% In contrast, PCDL sentence are satisfied by partial valuations.
% These partial views of the model are ordered, to denote the fact that some partial valuation may be more precise than another, i.e. provides truth values for more atomic propositions.
% In a concurrent program, this allows one thread to state conditions related to some local variables, while leaving the rest unconstrained.
%
In this paper, we study a class of distributive lattices that may be used as alternative models of partial observations.
This model generalises the PCDL model of~\cite{wbdkrs20} in two ways: we add the implication operator, and abstract over the atomic predicates of the language.
We intend this work to be used to define purpose-built models to tackle various verification or modelling tasks in concurrent settings.

To get intuitions about our model, consider a situation where a system is being monitored by an observer.
We are given a (possibly infinite) set of atomic propositions,
called observations in the following, that describe static properties of the system.
We also know which pairs of observations are coherent, i.e. may hold simultaneously, and which pairs are incompatible, meaning they are mutually exclusive.
The observer may report a set of observations that they witnessed.
This set need not be exhaustive: the fact that an atomic proposition does not appear in a report means that the observer failed to evaluate its veracity, not that they saw it to be false.
To establish the falsity of an observation $o$, the observer has to witness some other observation, incoherent with $o$.
This gives the model an intuitionistic flavour (existence of witness), as well as an epistemic one: knowledge may be partial, and is obtained by positive observations.
The experiment may be repeated, for instance observing each possible state of the system during an execution, or tying the observer to a program point and checking the system at this point for various runs with different inputs.
This way we obtain a set of reports, one for each experiment, each describing partially some possible state of the system.

Formally, this situation is modelled by an undirected graph, whose nodes are the atomic propositions, and whose adjacency relation is interpreted as coherence.
A single report, containing a set of pairwise coherent nodes, is simply a clique of the graph.
Non-exhaustivity is represented by an ordering between cliques.
We say that $c_2$ is more general than $c_1$ if $c_1$ contains $c_2$: every observation made in the report $c_2$ is also present in $c_1$, but the latter also deals with properties left unspecified by $c_2$.
A sentence in our model is then interpreted as a set of cliques such that if $c$ satisfies a formula, then any clique less general than $c$ also satisfies the same formula.
This way, we capture all partial states that are compatible with the report we received.

All the proofs in this paper have been formalised in Rocq, and are available online\footnote{\ttfamily\href{https://github.com/monstrencage/obs-alg-proofs/}{github.com/monstrencage/obs-alg-proofs/}}.

The paper is organised as follows. In Section~\ref{sec:def} we introduce observation algebras, and construct an implication such that the resulting structure is a Heyting algebra. In Section~\ref{sec:obs-lat} we present the free bounded distributive lattice in terms of observation algebras. The implication is added to the syntax in Section~\ref{sec:heyting-obs}, and two classes of graphs are investigated. For each we provide sound and complete axiomatisations of semantic contaiment. In Section~\ref{sec:prods} we further study a construction allowing one to combine graphs and their axiomatisations to produce a composite sound and complete system. This is leveraged to get interesting examples, representing memory states with infinitely many variables each with an arbitrary set of possible values. Finally, in Section~\ref{sec:remarks} we conclude with a few remarks on the Rocq formalisation, alternative semantics in terms of \emph{finite} cliques, and the matter of the empty clique.

\clearpage
\section{Definitions and first properties}
\label{sec:def}

%Generic results about algebras with hypotheses.

\subsection{Coherence graphs and observation algebras}
\label{sec:def-obs-alg}

\begin{definition}[Coherence graphs]
  A coherence graph is a pair $G=\tuple{\atobs(G),\coh(G)}$ of a set $\atobs(G)$ equipped with a symmetric and reflexive relation $\coh(G)$.
\end{definition}
The intuition here is that $\atobs(G)$ is a set of atomic observations.
The relation $\coh(G)$ indicates which pairs of observations may occur simultaneously.
If observations are seen as prediates over a set of events, then two predicates are coherent if they may both be true of the same event.
We will write $\incoh(G)$ for the complement of $\coh(G)$.
This relation is symmetric and irreflexive.

\begin{definition}[Coherence space]
  The coherence space generated by $G$ is its set of cliques, i.e.
  \begin{equation*}
    \Coh(G)\eqdef\setcompr{x\subseteq \atobs(G)}{\forall a,b\in x,\,a\coh(G) b}
    \qedhere
  \end{equation*}
\end{definition}

\begin{remark}
Note that coherence spaces, as introduced by Girard~\cite{girard_linear_1987}, can be equivalently described as a set $X$ of subsets of a base set $O$ satisfying the following closure properties:
\begin{align*}
  &\forall a,b\subseteq X,\,a\subseteq b\in X\Rightarrow a\in X
  &&\forall x\subseteq X,\,\paren{\forall a,b\in x,\,a\cup b\in X}\Rightarrow \bigcup x\in X.\qedhere
\end{align*}
\end{remark}

Following the same intuition as before, the coherence space records which sets of observations may describe a single event.
As such, if a clique $\beta$ is contained in a larger clique $\alpha$, then any event satisfying each of the predicates in $\alpha$ also satisfy those in $\beta$.
In this sense, $\beta$ is more general than $\alpha$, as its requirements are met by more events.
We formalise this generality/specificity relation:
\begin{definition}
  A clique $\alpha$ is less general, or more specific, than a clique $\beta$, written $\alpha\contains(G)\beta$, if $\beta\subseteq\alpha$.
\end{definition}
This could be summarised by stating $\contains(G)~=~\supseteq\cap\paren{\Coh(G)\times\Coh(G)}$.

To capture sets of events, it is natural to consider sets of cliques.
To push on the intuition that an event satisfies a clique if it satisfies each predicate in it, we restrict our attention to those sets $x$ of cliques that are down-closed with respect to the generality relation, i.e. $\alpha\contains(G)\beta\in x\Rightarrow\alpha\in x$.
This defines the observation algebra over $G$:
\begin{definition}[Observation algebra]
  The observation algebra over a coherence graph $G$ is the set of subsets of $\Coh(G)$ that are downclosed with respect to $\contains(G)$, i.e.:
  \begin{equation*}
    \Obs(G)\eqdef\setcompr{x\subseteq\Coh(G)}{\alpha\contains(G)\beta\in x\Rightarrow\alpha\in x}.
  \end{equation*}
  The elements of the observation algebra are called observations.
\end{definition}

It will be convenient to define the down-closure of an arbitrary set of cliques, i.e.
$$x\closure(G)\eqdef\setcompr{\alpha\in\Coh(G)}{\exists \beta\in x:\;\alpha\contains(G) \beta}.$$

\begin{notation}
  When the graph $G$ is obvious from the context, we write omit it from the notations, writing $\atobs$ instead of $\atobs(G)$, $\coh$ instead of $\coh(G)$, etc...
  We will let $a,b,c\dots$ range over $\atobs$, $\alpha,\beta,\gamma,\dots$ range over $\Coh$, and $x,y,z,\dots$ range over $\Obs$.
  Finally, we denote the subset of $\Coh$ containing only the finite cliques of $G$ by the notation $\Cohf$.
\end{notation}

With this operator, we may reformulate the definition of $\Obs$ as follows:
\begin{equation}
  \label{eq:alt-def-Obs}
  \Obs=\setcompr{x\closure}{x\subseteq\Coh}.
  \tag{Alternative definition of $\Obs$}
\end{equation}

This operator makes the following result straightforward to establish:
\begin{proposition}\label{prop:obs-bdl}
  $\tuple{\Obs,\cup,\cap,\emptyset,\Coh}$ is a bounded distributive lattice.
\end{proposition}
\begin{proof}
  Since we know that $\tuple{\pset{\Coh},\cup,\cap,\emptyset,\Coh}$ is a bounded distributive lattice, we only need to prove that both constants are observations, and that $\cup$ and $\cap$ are internal to $\Obs$. This only relies on the fact that $\contains$ is a preorder and on the alternative definition of $\Obs$.

  Indeed, $\contains$ being a preorder entails that $\closure$ is a Kuratowski closure operator~\cite{kuratowski1922}, i.e.:
  \begin{align}
    \paren{x\cup y}\closure&={x\closure}\cup {y\closure}\label{eq:14}\\
    x&\subseteq x\closure\label{eq:22}\\
    \emptyset&=\emptyset\closure\label{eq:23}\\
    x\closure\closure&=x\closure\label{eq:24}
  \end{align}
  From these we get:
  \begin{align*}
    &\emptyset=\emptyset\closure\in\Obs\\
    &\paren{{x\closure}\cup {y\closure}}\closure=
  \paren{x\closure\closure}\cup \paren{y\closure\closure}=
  {x\closure}\cup {y\closure}
  \end{align*}
  Using the fact that any subset of $\Coh$ is contained in $\Coh$, we also get:
  \begin{align*}
    &\Coh\subseteq\Coh\closure
    \text{ and }\Coh\closure\subseteq\Coh
    &&\Rightarrow
    &&\Coh=\Coh\closure\in\Obs
  \end{align*}
  Finally, Kuratowski closures have the property that $\paren{x\cap y}\closure\subseteq {x\closure}\cap {y\closure}$~\cite[Theorem 2]{kuratowski1922}, hence:
  \begin{align*}
    &\paren{{x\closure}\cap{ y\closure}}\closure\subseteq
  \paren{x\closure\closure}\cap \paren{y\closure\closure}=
      {x\closure}\cap {y\closure}\\
    &{x\closure}\cap {y\closure}\subseteq \paren{{x\closure}\cap {y\closure}}\closure\tag*{by~\eqref{eq:22}}
  \end{align*}
  Therefore $ {x\closure}\cap {y\closure}=\paren{{x\closure}\cap {y\closure}}\closure\in\Obs$.
\end{proof}

We will use sometimes the fact that while $\emptyset\closure=\emptyset$, the empty set of atomic observations is a clique, and furthermore satisfies $\set\emptyset\closure=\Coh$.

\subsection{Observation algebras are Heyting algebras}
\newcommand\impl{\mathbin{\rightarrow}}
Observation algebras may be equipped with an implication, in the sense of Heyting algebras.
\begin{equation*}
  \label{eq:impl}
  x\impl y \eqdef\setcompr{\alpha\in \Coh}{\forall \beta\in x,\,\alpha\cup\beta\in\Coh\Rightarrow \alpha\cup\beta\in y}
\end{equation*}
Notice that this operator is antitone in its first argument (because of the universal quantification on $\beta\in x$), and monotone in its second argument.
This operation is internal to the set of observations, as witnessed by the following identity:
\begin{equation}
  \label{eq:impl-closure}
   {x \closure}\impl{y\closure} = \paren{{x\closure}\impl {y\closure}}\closure = x \impl {y\closure}
\end{equation}
\begin{proof}
  We know by properties of the closure that ${x \closure}\impl{y\closure}\subseteq\paren{{x\closure}\impl {y\closure}}\closure $.
  To prove the two identities, it is therefore enough to check that $\paren{{x\closure}\impl {y\closure}}\closure \subseteq x \impl {y\closure}\subseteq {x \closure}\impl{y\closure}$.
  We now establish these two containments.
  
  Assume $\alpha \in  \paren{{x\closure}\impl {y\closure}}\closure$. By definition there is $\alpha'\in\Coh$ such that $\alpha\contains \alpha'$ and for any $\beta\in {x\closure}$ such that $\alpha'\cup\beta\in\Coh$ we have $\alpha'\cup\beta\in {y\closure}$. Let $\beta\in x$ be such that $\alpha\cup\beta\in\Coh$. Since $\alpha\contains \alpha'$, we have $\alpha'\cup\beta\subseteq\alpha\cup\beta\in\Coh$, which entails $\alpha'\cup\beta\in\Coh$. Since $\beta\in x\subseteq {x\closure}$, by hypothesis we get $\alpha'\cup\beta\in {y\closure}$. Since $\alpha\cup\beta\contains\alpha'\cup\beta\in {y\closure}$ we obtain $\alpha\cup\beta\in{y\closure\closure}=y\closure$. This proves $\alpha\in x\impl{y\closure}$, hence that $\paren{{x\closure}\impl {y\closure}}\closure\subseteq x\impl{y\closure}$.

  For the converse inclusion, let $\alpha$ be in $x\impl{y\closure}$, and let $\beta\in x\closure$ such that $\alpha\cup\beta\in \Coh$. By definition there is some $\beta'\subseteq\beta$ such that $\beta'\in x$. Since $\alpha\cup\beta'\subseteq\alpha\cup\beta\in\Coh$, we get $\alpha\cup\beta'\in \Coh$, hence by hypothesis $\alpha\cup\beta'\in{y\closure}$. As such we obtain $\alpha\cup\beta\contains\alpha\cup\beta'\in y\closure $ we have $\alpha\cup\beta\in {y\closure\closure}=y\closure$, meaning $\alpha\in{x\closure}\impl{y\closure}$.
  This proves $x\impl{y\closure}\subseteq{x\closure}\impl{y\closure}$.
\end{proof}
\begin{proposition}\label{prop:heyting}
  $\tuple{\Obs,\cup,\cap,\emptyset,\Coh,\impl}$ is a Heyting algebra.
\end{proposition}
\begin{proof}
  Since we already know that $\Obs$ is a bounded lattice, we only need to show that $x\cap y\subseteq z\Leftrightarrow x\subseteq y\impl z$.
  
  For the right-to-left implication, if $\alpha\in x\cap y$, then $\alpha\in x\subseteq y\impl z$ and $\alpha\in y$, so by definition of $\impl$ we get $\alpha=\alpha\cup\alpha\in z$.

  For the direct implication, assume $\alpha\in x$. To prove that $\alpha\in y\impl z$, consider an arbitrary $\beta\in y$ such that $\alpha\cup\beta\in \Coh$.
  Since $x$ and $y$ are observations, and since $\alpha\cup\beta$ is more specific than both $\alpha$ and $\beta$, we know that $\alpha\cup\beta\in x\cap y\subseteq z$.
  By definition of $\impl$, this entails that $\alpha\in y\impl z$.
\end{proof}

As any Heyting algebra, $\Obs$ has a pseudocomplement, defined as $\overline x\eqdef x\impl\bot$.
In other words, we have:
$$ \alpha\in\overline x\Leftrightarrow\forall \beta\in x,\,\alpha\cup\beta\notin\Coh.$$

\clearpage
\section{Observation lattices}
\label{sec:obs-lat}

In this section, we provide a syntax to describe observations, and use it to give an axiomatic presentation of observation lattices. We fix a graph $G$ for the remainder of the section, and let $a,b,\dots$ range over its vertices.
The syntax in question consists of terms built out of the lattice operators and constants, as well as atomic observations:
\begin{align*}
  s,t\in \Tlat&\Coloneqq a \Mid s\wedge t\Mid s\vee t\Mid\top\Mid\bot.
\end{align*}
These terms may be interpreted as observations in a straightforward fashion :
\begin{align*}
  \sem a&\eqdef\set {\set a}\closure
  &\sem{\top}&\eqdef\Coh
  &\sem \bot&\eqdef\emptyset\\
  \sem {s\vee t}&\eqdef\sem s\cup\sem t
  &\sem {s\wedge t}&\eqdef\sem s\cap\sem t
  % &\sem {s\impl t}&\eqdef\sem s\impl\sem t
\end{align*}

\subsection{Bounded distributive lattices}
\label{sec:bnd-distr-lat}

As a first step, we define the axiomatic relation $\equiv_0$ over $\Tlat$ as the smallest congruence satisfying the axioms~\eqref{ax:1} to~\eqref{ax:8}, i.e. the axioms of bounded distributive lattices.
As usual for such lattices, we may also define a preorder $\leqq_0$ by setting $s\leqq_0t\Leftrightarrow s\vee t\equiv_0 t$.
This ordering is a partial order up-to $\equiv_0$ (that is $s\leqq_0 t\leqq_0s$ entails $s\equiv_0 t$), and $s\leqq_0t$ is equivalent to $s\wedge t\equiv_0s$.

\begin{table}[h]
  \noindent%
  \begin{tableequations}
    \begin{minipage}[t]{.33\linewidth}
      \begin{align}
        &&s\vee\paren{t\vee u}
        &\equiv \paren{s\vee t}\vee u\axlabel{ax:1}\\
        &&s\vee t&\equiv t\vee s\axlabel{ax:2}\\
        &&s\vee\bot&\equiv s\axlabel{ax:3}
      \end{align}
    \end{minipage}\hfill%
    \begin{minipage}[t]{.33\linewidth}
      \begin{align}
        &&s\wedge\paren{t\wedge u}&\equiv \paren{s\wedge t}\wedge u\axlabel{ax:4}\\
        &&s\wedge t&\equiv t\wedge s\axlabel{ax:5}\\
        &&s\wedge\top&\equiv s\axlabel{ax:6}
      \end{align}
    \end{minipage}\hfill%
    \begin{minipage}[t]{.33\linewidth}
      \begin{align*}
        &&s\wedge\paren{s\vee t}&\equiv s\axlabel{ax:7}\\
        &&s\vee\paren{s\wedge t}&\equiv s\axlabel{ax:8}
      \end{align*}
    \end{minipage}
  \end{tableequations}
  \caption{Axioms of bounded distributive lattices}
  \label{tab:axioms}
\end{table}

Since $\tuple{\Tlat,\equiv_0,\vee,\wedge,\bot,\top}$ is a bounded distributive lattice, we may define operators $\bigvee,\bigwedge:\fpset\Tlat\to\Tlat$ such that:
\begin{align}
  &&\bigvee\emptyset&=\bot
  &\bigwedge\emptyset&=\top\\
  \forall S,T\in\fpset{\fpset \atobs},
  &&\bigvee \paren{S\cup T}&\equiv_0\bigvee S\vee\bigvee T
  &\bigwedge \paren{S\cup T}&\equiv_0\bigwedge S\wedge\bigwedge T
\end{align}

A handy tool in our completeness proofs will be the following interpretation of lattice terms:
\begin{align*}
  \brack a&\eqdef\set {\set a}
  &\brack{\top}&\eqdef\set\emptyset
  &\brack \bot&\eqdef\emptyset\\
  \brack {s\vee t}&\eqdef\brack s\cup\brack t
  &\brack {s\wedge t}&\eqdef\setcompr{x\cup y}{x\in\brack s,\;y\in\brack t}
\end{align*}
Notice that the image of a lattice term is always a finte set of finite sets of atomic observations.

By using the classic algorithm to normalise such terms into disjunctive normal forms, we may easily prove the following universal identity:
\begin{align}
  s\equiv_0\bigvee\setcompr{\bigwedge x}{x\in\brack s}\label{eq:dnf}
\end{align}
(In this statement and elsewhere, we implicitely use the fact that $\atobs\subseteq\Tlat$).

These finite sets should be considered up-to pairwise containment, i.e. the partial order $\lhd$ defined as :
\[\forall X,Y\in\fpset{\fpset\atobs},\;X\lhd Y\Leftrightarrow \forall x\in X,\,\exists y\in Y:\;x\supseteq y.\]
Indeed, we may prove by a simple induction on derivations that
\begin{align}
  s\leqq_0 t\Rightarrow \brack s\lhd\brack t.\label{eq:lat-sound}
\end{align}
We may also prove (by induction on the size of $X$) that for any two finite sets of finite sets of atomic observations $X,Y$ we have:
\begin{align}
  X\lhd Y\Leftrightarrow \bigvee\setcompr{\bigwedge x}{x\in X}\leqq_0\bigvee\setcompr{\bigwedge y}{y\in Y}.\label{eq:lhd-ax}
\end{align}

Taken together, equations \eqref{eq:dnf} and~\eqref{eq:lhd-ax} may be used to establish the following correspondance:
\begin{align}
  s\leqq_0t\Leftrightarrow \brack s\lhd\brack t.
\end{align}

\subsection{Observation lattices}
\label{sec:obs-lat-compl}

To move from this simple interpretation to the interpretation in terms of observations, we observe the following correspondance:
\begin{lemma}\label{lem:brack-to-sem}
  For any lattice term $s\in\Tlat$, we have $\sem s=\paren{\brack s\cap\Cohf}\closure$.
\end{lemma}
\begin{proof}
  By induction on $s$:
  \begin{itemize}
  \item Since $\set a\in\Cohf$, it holds that
    $\sem a=\set{\set a}\closure=\paren{\set{\set a}\cap\Cohf}\closure=\paren{\brack a\cap\Cohf}\closure$.
  \item Since $\emptyset\in\Cohf$, and since $\emptyset$ is maximal in $\tuple{\Coh,\contains}$, it holds that
    $$\sem \top=\Coh=\set{\emptyset}\closure=\paren{\set{\emptyset}\cap\Cohf}\closure=\paren{\brack \top\cap\Cohf}\closure.$$
  \item Trivially, $\sem \bot=\emptyset=\emptyset\closure=\paren{\emptyset\cap\Cohf}\closure=\paren{\brack\bot\cap\Cohf}\closure$.
  \item For the case of disjunction, by induction hypothesis we assume $\sem s=\paren{\brack s\cap\Cohf}\closure$ and $\sem t=\paren{\brack t\cap\Cohf}\closure$.
    This case can be dispatched easily thanks to the properties of Kuratowski closures:
    \begin{align*}
      \sem {s\vee t}
      =\sem s\cup\sem t
      &=\paren{\brack s\cap\Cohf}\closure\cup\paren{\brack t\cap\Cohf}\closure\\
      &=\paren{\paren{\brack s\cap\Cohf}\cup\paren{\brack t\cap\Cohf}}\closure
      =\paren{\paren{\brack s\cup\brack t}\cap\Cohf}\closure
      =\paren{\brack {s\vee t}\cap\Cohf}\closure
    \end{align*}
  \item For the conjunction, by induction hypothesis we assume the following: $$\sem s=\paren{\brack s\cap\Cohf}\closure\text{ and }\sem t=\paren{\brack t\cap\Cohf}\closure.$$
    Therefore we have:
    \[\sem {s\wedge t}
      =\sem s\cap\sem t
      =\paren{\brack s\cap\Cohf}\closure\cap\paren{\brack t\cap\Cohf}\closure\]
    For this case we proceed by mutual contaiment.
    \begin{itemize}
    \item Assume $\alpha\in\sem {s\wedge t}=\paren{\brack s\cap\Cohf}\closure\cap\paren{\brack t\cap\Cohf}\closure$.
      There are $\alpha_1\in\brack s\cap\Cohf$ and $\alpha_2\in\brack t\cap\Cohf$ such that $\alpha\contains\alpha_1$ and $\alpha\contains\alpha_2$.
      Since $\contains=\supseteq$, this means that $\alpha\contains\alpha_1\cup\alpha_2$.
      Furthermore, $\alpha_1\in\brack s$ and $\alpha_2\in\brack t$ entail that
      $\alpha_1\cup\alpha_2\in\brack{s\wedge t}$.
      Both $\alpha_1$ and $\alpha_2$ being finite, so is their union.
      Finally, since $\alpha\in\Coh$ and $\alpha\contains\alpha_1\cup\alpha_2$, the latter is a clique.
      To sum up, we have $\alpha\contains\alpha_1\cup\alpha_2\in\brack{s\wedge t}\cap\Cohf$, hence $\alpha\in\paren{\brack {s\wedge t}\cap\Cohf}\closure$.
    \item Assume $\alpha\in\paren{\brack {s\wedge t}\cap\Cohf}\closure$,
     i.e. there are  $\alpha_1\in\brack s$ and $\alpha_2\in\brack t$ such that $\alpha_1\cup\alpha_2\in\Cohf$ and $\alpha\contains \alpha_1\cup\alpha_2$.
      The fact that $\alpha_1\cup\alpha_2\in\Cohf$ ensures that both $\alpha_1$ and $\alpha_2$ are finite cliques.
      Therefore, we have $\alpha\contains\alpha_1\in\brack s\cap\Cohf$ and $\alpha\contains\alpha_2\in\brack t\cap\Cohf$, thus proving
      $\alpha\in\paren{\brack s\cap\Cohf}\closure\cap\paren{\brack t\cap\Cohf}\closure=\sem {s\wedge t}$.\qedhere
    \end{itemize}
  \end{itemize}
\end{proof}

We will use this lemma in the form of the following corollary:
\begin{corollary}\label{cor:brack-to-sem}
  For any pair of terms $s,t\in\Tlat$, we have $\sem s\subseteq\sem t$ if and only if $\brack s\cap\Cohf\lhd\brack t$.
\end{corollary}
\begin{proof}
  Thanks to \Cref{lem:brack-to-sem},  we have $\sem s\subseteq\sem t$ if and only if $\paren{\brack s\cap\Cohf}\closure\subseteq\paren{\brack t\cap\Cohf}\closure$.
  The operator $\closure$ being a Kuratowski closure, this is futhermore equivalent to $\brack s\cap\Cohf\subseteq\paren{\brack t\cap\Cohf}\closure$.
  So, unfolding the definitions, the containment of interest is equivalent to the following property:
  \begin{align*}
    \forall \alpha\in\brack s\cap\Cohf,\;\exists \beta\in\brack t\cap\Cohf:\;
    \alpha\contains\beta
  \end{align*}
  Observe that if $\alpha\in\Cohf$ and $\alpha\contains \beta$, i.e. $\beta\subseteq\alpha$, $\beta\in\Cohf$.
  This means the property is equivalent to the following:
  \begin{align*}
  \forall \alpha\in\brack s\cap\Cohf,\;\exists \beta\in\brack t:\;
    \alpha\supseteq\beta
  \end{align*}
  This happens to be the definition of the statement $\brack s\cap\Cohf\lhd\brack t$.
\end{proof}

We now define the axiomatic relation $\equiv_1$ on $\Tlat$, as the smallest congruence satisfying axioms~\eqref{ax:1} to~\eqref{ax:8}, as well as the following axiom.
\begin{tableequations}
  \begin{align}
    \forall a\incoh b,
    && a\wedge b&\equiv\bot\axlabel{ax:13}
  \end{align}
\end{tableequations}
This relation contains $\equiv_0$, and since $\Obs$ is a bounded distributive lattice and $$\sem{a\wedge b}=\setcompr{\alpha\in\Coh}{a,b\in\alpha}=\emptyset$$ for any incoherrent pair $a\incoh b$, it is clearly sound with respect to the interpretation $\sem -$.

To establish completeness, we only need to following lemma:
\begin{lemma}\label{lem:not-clique-bot}
  For any finite set of atomic observations $x\in\fpset\atobs\setminus \Cohf$, we have $\bigwedge x\equiv_1\bot$.
\end{lemma}
\begin{proof}
  $x\in\fpset\atobs\setminus \Cohf$ means that there are $a,b\in x$ such that $a\incoh b$.
  As such, we have $x\supseteq\set{a,b}$, meaning
  \begin{align*}
    &\bigwedge x\leqq_0\bigwedge\set{a,b}\equiv_0a\wedge b\equiv_1\bot.\qedhere
  \end{align*}
\end{proof}
Since lemma enables us to adapt \eqref{eq:dnf} to a clique-based semantics:
\begin{corollary}\label{cor:dnf-cliques}
  For any $s\in\Tlat$, we have $s\equiv_1\bigvee\setcompr{\bigwedge\alpha}{\alpha\in\brack s\cap\Cohf}$.
\end{corollary}
\begin{proof}
  \begin{align*}
    s&\equiv_0\bigvee\setcompr{\bigwedge x}{x\in\brack s}
    \tag*{by~\eqref{eq:dnf}}\\
    &=\bigvee\setcompr{\bigwedge x}{x\in\paren{\brack s\cap\Cohf}\cup\paren{\brack s\setminus\Cohf}}\\
    &=\bigvee\paren{\setcompr{\bigwedge x}{x\in\brack s\cap\Cohf}\cup\setcompr{\bigwedge x}{x\in\brack s\setminus\Cohf}}\\
    &\equiv_0\paren{\bigvee\setcompr{\bigwedge x}{x\in\brack s\cap\Cohf}}\vee\paren{\bigvee\setcompr{\bigwedge x}{x\in\brack s\setminus\Cohf}}\\
    &\equiv_1\paren{\bigvee\setcompr{\bigwedge x}{x\in\brack s\cap\Cohf}}\vee\paren{\bigvee\setcompr{\bot}{x\in\brack s\setminus\Cohf}}
    \tag*{by \Cref{lem:not-clique-bot}}\\
    &\equiv_0 \paren{\bigvee\setcompr{\bigwedge x}{x\in\brack s\cap\Cohf}}\vee\bot\equiv_0\bigvee\setcompr{\bigwedge x}{x\in\brack s\cap\Cohf}.\qedhere
  \end{align*}
\end{proof}
\begin{theorem}\label{thm:obs-lat-compl}
  For any pair $s,t\in\Tlat$, we have $\sem s\subseteq \sem t$ if and only if $s\leqq_1 t$. 
\end{theorem}
\begin{proof}
  The right-to-left implication holds by soundness, i.e. the fact that the axioms defining $\equiv_1$ all hold in $\Obs$.
  
  For the other implication, assume $\sem s\subseteq\sem t$.
  We first apply \Cref{cor:brack-to-sem} to obtain that $\brack s\cap\Cohf\lhd\brack t$.
  By~\eqref{eq:lhd-ax} this is equivalent to
  \begin{equation}
    \bigvee\setcompr{\bigwedge x}{x\in\brack s\cap\Cohf}
    \leqq_0\bigvee\setcompr{\bigwedge y}{y\in\brack t}\eqproof{2}
  \end{equation}
  Hence we get:
  \begin{align*}
    s&\equiv_1\bigvee\setcompr{\bigwedge\alpha}{\alpha\in\brack s\cap\Cohf}
       \tag*{by \Cref{cor:dnf-cliques}}\\
     &\leqq_0\bigvee\setcompr{\bigwedge y}{y\in\brack t}
    \tag*{by \eqrefproof{2}}\\
     &\equiv_0t
       \tag*{by \eqref{eq:dnf}}
       % \qedhere
  \end{align*}
  Summing up, since $\equiv_0~\subseteq~\equiv_1$, we obtain $s\leqq_1t$. 
\end{proof}

\clearpage
\section{(Heyting) observation algebras}
\label{sec:heyting-obs}

\subsection{For arbitrary graphs}
We now add the implication to the signature, considering $\Tobs$ as our set of terms.
\begin{align*}
  s,t\in\Tobs&\Coloneqq a \Mid s\wedge t\Mid s\vee t\Mid  s\impl t\Mid \top\Mid \bot
\end{align*}
The interpretation $\sem -$ is extended to $\Tobs$ by setting
\[\sem{s\impl t}\eqdef\sem s\impl \sem t=\setcompr{\alpha\in\Coh}{\forall \beta\in\sem s,\,\alpha\cup\beta\in\Coh\Rightarrow \alpha\cup\beta\in\sem t}.\]
As we have done in previous section, we want to define an axiomatic equivalence relation that precisely captures semantic equivalence of terms.
Let us consider which axioms should be chosen to achieve this.

As we showed in~\Cref{prop:heyting}, $\Obs$ is a Heyting algebra.
It then makes sense to include axioms~\eqref{ax:1}-\eqref{ax:8} and~\eqref{ax:9}-\eqref{ax:12}, which define Heyting algebras.

\begin{tableequations}
  \noindent%
  \begin{minipage}{.49\linewidth}
    \begin{align}
      &&s\impl s&\equiv \top\axlabel{ax:9}\\
      &&s\wedge\paren{s\impl t}&\equiv s\wedge t\axlabel{ax:10}
    \end{align}
  \end{minipage}\hfill
  \begin{minipage}{.49\linewidth}
    \begin{align}
      &&t\wedge\paren{s\impl t}&\equiv t\axlabel{ax:11}\\
      &&s\impl\paren{t\wedge u}
      &\equiv\paren{s\impl t}\wedge\paren{s\impl u}\axlabel{ax:12}
    \end{align}
  \end{minipage}
\end{tableequations}

However, this cannot be enough, as our observation algebra is not the free Heyting algebra: there are additionnal identities resulting from the coherence graph structure.
In particular, we observe that formulas of the shape $\bigwedge \alpha\impl s$, where $\alpha\in\Cohf$, satisfy some distributivity laws that do not hold for general formulas.
Indeed, for any observations $x,y\in\Obs$, we have that $\set\alpha\impl\paren{x\cup y}=\paren{\set\alpha\impl x}\cup\paren{\set \alpha\impl y}$.
Furthermore, for any atomic observation $a\in\atobs$ such that $a\notin\alpha$, we can see that $\beta\in\set\alpha\impl\set{\set a}\closure$ if and only if either $\beta$ and $\alpha$ are incompatible (i.e. $\beta\cup\alpha\notin\Coh$), or if $a\in\beta\cup\alpha$, which amounts to $a\in\beta$ since we assumed $a\notin\alpha$.
These observations is summarised in the following axioms:
\begin{tableequations}
  \begin{align}
    \forall \alpha\in\Cohf,&&\bigwedge \alpha\impl\paren{s\vee t}
    &\equiv\paren{\bigwedge \alpha\impl s}\vee\paren{\bigwedge \alpha\impl t}
      \axlabel{ax:14}\\
    \forall \alpha\in\Cohf,\,
    \forall a\notin\alpha,
                           &&\bigwedge\alpha\impl a
    &\equiv \paren{\bigwedge\alpha\impl\bot}\vee a.\axlabel{ax:15}
  \end{align}
\end{tableequations}
\begin{remark}
  This pair of axioms holds more widely than stated here.
  In \eqref{ax:14} $\alpha$ does not need to be coherent: the axiom holds for any finite set of atomic observations (but not for arbitrary formulas).
  Similary, the same extension works for \eqref{ax:15} (i.e. $\alpha\in\fpset\atobs$ and $a\notin\alpha$).
\end{remark}

We denote by $\equiv_2$ the least congruence containing the axioms we have stated so far, i.e. \eqref{ax:1}-\eqref{ax:8},~\eqref{ax:13},~\eqref{ax:9}-\eqref{ax:12}, and~\eqref{ax:14}-\eqref{ax:15}.
This relation is sound, since each of these axioms is valid in observation algebras.
However, it fails to be complete, as we shall now demonstrate.

\begin{proposition}
  There are true semantic identities not derivable from the axioms given so far.
\end{proposition}
\begin{proof}
  Given a coherence graphs $H$, an induced subgraph is a graph $G$ such that $\atobs(G)\subseteq\atobs(H)$, and such that for any $a,b\in\atobs(G)$ we have $a\coh(F)b\Leftrightarrow a\coh(H)b$.
  It is straightforward to check that in such a case, for any terms $s,t\in\Tobs(G)$ we have:
  \[s\equiv^G_2t\Rightarrow s\equiv^H_2 t.\]
  Let us define two graphs : $G=\tuple{\set \circ,id_{\set\circ}}$ and $H=\tuple{\set{\circ,\bullet},id_{\set{\circ,\bullet}}}$.
  We observe the following:
  \begin{enumerate}
  \item $G$ is an induced sub-graph of $H$.
  \item $\sem{\bullet\impl\bot}_G=\emptyset=\sem{\bot}_G$.
  \item $\sem{\bullet\impl\bot}_H=\set{\set\bullet}=\sem{\bullet}_H$.
  \end{enumerate}
  So, if $\equiv^G_2$ were complete, then $\bullet\impl\bot\equiv^G_2\bot$ would hold.
  $G$ being an induced subgraph of $H$, this would entail $\bullet\impl\bot\equiv^H_2\bot$.
  By soundness of $\sem-_H$, this leads to the a contradiction since $\sem{\bullet\impl\bot}_H=\set{\set\bullet}\neq\emptyset=\sem{\bot}_H$.
\end{proof}

In order to obtain complete axiomatisations, we will restrict our attention to certain classes of graphs.
Two of these will be described in the remainder of this section.
We will also show in the next section a construction to build new tractable graphs (and the corresponding axiomatisations) out of smaller tractable graphs.

Since we will perform several completeness proofs, we start by sketching a general strategy, which we will then instanciate each time.
This strategy may be seen as a generalisation of the notion of normalisation procedure, but where the normalisation is explicitely factored into 1) a translation $\tau$ of terms into an intermediary data structure; 2) a translation $\phi$ of such data structures into terms.
This is more general than normalisations, because no unicity or injectivity requirement is necessary.
It could be recast as a normalisation procedure if one was to move from sets and functions to setoids and equivalence preserving functions.
\begin{definition}
Given a graph $G$ and an equivalence relation $\equiv$ over $\Tobs(G)$, a representation is a triple $\rho=\tuple{R,\tau,\phi}$ such that:
\begin{itemize}
\item $R$ is a set of representatives, with translations $\tau:\Tobs(G)\to R$ and $\phi:R\to\Tobs(G)$;
\item for any term $s\in\Tobs(G)$, we have $\phi(\tau(s))\equiv s$;
\item for any pair of representatives $r,r'\in R$, if $\sem{\phi(r)}=\sem{\phi(r')}$, then $\phi(r)\equiv\phi(r')$.\qedhere
\end{itemize}
\end{definition}
\begin{lemma}\label{lem:repr-impl-compl}
  If there exists a representation of $G$ with respect to some sound relation $\equiv$, then $\equiv$ is complete for $G$.
\end{lemma}
\begin{proof}
  Let $\sem s=\sem t$. Since $\phi(\tau(s))\equiv s$ and $\phi(\tau(t))\equiv t$, and we assumed $\equiv$ to be sound, $\sem{\phi(\tau(s))}=\sem{\phi(\tau(t)}$.
  Therefore, $\phi(\tau(s))\equiv \phi(\tau(t)$ from which we can derive $s\equiv\phi(\tau(s))\equiv \phi(\tau(t)\equiv t$.
\end{proof}
\begin{remark}[On decidability]
  Representation triples will also be useful to obtain decidability results.
  Indeed, if we have a representation triple, deciding whether $s\equiv t$ may be achieved by comparing $\sem{\phi\circ\tau(s)}$ and $\sem{\phi\circ\tau(t)}$.
  The later tends to be easier in practice, since we only need to compare ``normalised'' terms. 
\end{remark}
\subsection{Graphs with finite anti-neighbourhoods}
\label{sec:fan-graphs}

Let us introduce the \emph{finite anti-neighbourhood} (a.k.a. FAN) property.
\begin{definition}[FAN]
  A graph is said to have the \emph{finite anti-neighbourhood} property if for any node $a\in\atobs$, the set $\setcompr{b\in\atobs}{a\incoh b}$ is finite.
\end{definition}
Note that any finite graph has the FAN property, as do infinite cliques.

In the remainer of this section, we assume the coherence graph under consideration to have the FAN property.
This allows us to state axiom~\eqref{ax:16}, i.e. the fact that for any finite clique $\alpha\in\Cohf$, the negation of $\alpha$ is the (downclosure of the) set of atomic observations incoherent with some atom in $\alpha$.
\begin{tableequations}
  \begin{align}
    \forall \alpha\in\Cohf,&&\paren{\bigwedge\alpha}\impl\bot
    &\equiv\bigvee\setcompr{a\in\atobs}{\exists b\in\alpha:\,a\incoh b}
      \axlabel{ax:16}
  \end{align}
\end{tableequations}

\begin{figure}[h!]
  \centering
  \begin{tikzpicture}
    \node(0) at (2,1) {0};
    \node(1) at (3,1) {1};
    \node(2) at (4,1) {2};
    \node(n) at (6,1) {$n$};
    \node at ($(2)!.5!(n)$){$\cdots$};
    \node at ($(8,1)!.5!(n)$){$\cdots$};
    \node(b) at (1,0.5) {$\bullet$};
    \node(w) at (1,1.5) {$\circ$};
    \draw[thick] (b) -- (w)
    (0)--(1)
    (1)--(2)
    (0) to[bend right](2)
    (0) to[bend left](n)
    (1) to[bend left](n)
    (2) to[bend left](n);
  \end{tikzpicture}
  \caption{Example of graph without the FAN property}
  \label{fig:ex-no-fan}
\end{figure}
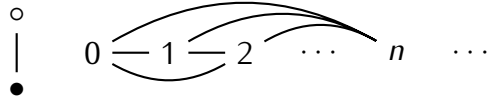
\begin{remark}
  Without the FAN property, the set $\setcompr{a\in O}{\exists b\in \alpha:a\incoh b}$ might be infinite, thus preventing the expression of this law in our syntax.
  This could lead to some unprovable identities.
  Consider for instance the case of a graph represented in~\Cref{fig:ex-no-fan}, with vertices $\Nat\cup\set{\bullet,\circ}$, and where $\bullet\incoh n$ and $\circ\incoh n$ for every $n\in\Nat$, but $\bullet\coh\circ$, and for any two $n,m\in\Nat$ we have $n\coh m$.
  Clearly we have $\set\bullet\impl\emptyset=\set\circ\impl\emptyset=\pset\Nat\setminus\set\emptyset$.
  However, it is highly doubtful this identity would be provable in a finitely presentable system, since its validity relies on an infinite set of witnesses (i.e. $\Nat$).
\end{remark}

We now define the axiomatic equivalence relation $\equiv_{\text{fan}}$ as the smallest congruence on $\Tobs$ satisfying axioms \eqref{ax:1}-\eqref{ax:8}, \eqref{ax:9}-\eqref{ax:12}, \eqref{ax:14}-\eqref{ax:15}, as well as axiom~\eqref{ax:16}.
As we argued already, all of these axioms hold in $\Obs$, meaning $\equiv_{\text{fan}}$ is sound with respect to $\sem -$.
It is straightfoward to show that axiom \eqref{ax:13} is derivable from the others, so we have for $s,t\in\Tlat\paren{\subseteq\Tobs}$ that $s\equiv_2t\Rightarrow s\equiv_{\text{fan}}t$.

To prove completeness, we will build a representation triple, with $R=\Tlat$.
The translation $\phi$ from $R$ to $\Tobs$-terms is simply the identity, since $\Tlat\subseteq\Tobs$.
As such, whenever $\sem{\phi(r)}=\sem{\phi(r')}$, by \Cref{thm:obs-lat-compl} we get $\phi(r)\equiv_1\phi(r')$.
The axioms of observation lattices being either included (in the case of \eqref{ax:1}-\eqref{ax:8}) or derivable (for \eqref{ax:13}) from those of $\equiv_{\text{fan}}$, this entails $\phi(r)\equiv_{\text{fan}}\phi(r')$.

What remains to be done is the construction of $\tau$ from $\Tobs$ to $\Tlat$ such that $s\equiv_{\text{fan}}\phi(\tau(s))=\tau(s)$.
We define it by induction, the only non-trivial case being that of the implication, which we detail in the following lemma.

\begin{lemma}\label{lem:impl-to-lat}
  For any pair of terms $s,t\in\Tlat$, there is a term $I_s^t\in\Tlat$ such that $s\impl t\equiv_{\text{fan}}I_s^t$.
\end{lemma}
\begin{proof}
  We start by putting both terms in disjunctive normal form, with \Cref{cor:dnf-cliques}.
  We get:
  \[s\impl t\equiv_{\text{fan}}\paren{\bigvee\setcompr{\bigwedge \alpha}{\alpha\in\brack s\cap\Cohf}}\impl\paren{\bigvee\setcompr{\bigwedge \beta}{\beta\in\brack t\cap\Cohf}}\]  
  
  It is well known (and easily derivable) that in any Heyting algebra we have $\paren{s\vee t}\impl u\equiv\paren{s\impl u}\wedge\paren{t\impl u}$.
  By rewriting this law to the right hand side we get:
  \[s\impl t\equiv_{\text{fan}}\bigwedge\setcompr{\bigwedge \alpha\impl\paren{\bigvee\setcompr{\bigwedge \beta}{\beta\in\brack t\cap\Cohf}}}{\alpha\in\brack s\cap\Cohf}\]

  If $\brack t\cap\Cohf=\emptyset$, this expression can be turned into a $\Tlat$ term by~\eqref{ax:16}:
  \begin{align*}
    s\impl t
    &\equiv_{\text{fan}}\bigwedge\setcompr{\bigwedge \alpha\impl\paren{\bigvee\emptyset}}{\alpha\in\brack s\cap\Cohf}\\
    &=\bigwedge\setcompr{\bigwedge \alpha\impl\bot}{\alpha\in\brack s\cap\Cohf}\\
    &\equiv_{\text{fan}}\bigwedge\setcompr{\bigvee\setcompr{b}{\exists a\in\alpha:\,a\incoh b}}{\alpha\in\brack s\cap\Cohf}
  \end{align*}

  If on the other hand $\brack t\cap\Cohf\neq\emptyset$, we may rewrite \eqref{ax:14} to obtain:
  \[s\impl t\equiv_{\text{fan}}\bigwedge\setcompr{\bigvee\setcompr{\bigwedge \alpha\impl\bigwedge \beta}{\beta\in\brack t\cap\Cohf}}{\alpha\in\brack s\cap\Cohf}\]
  With \eqref{ax:12} we get:
  \[s\impl t\equiv_{\text{fan}}
    \bigwedge\setcompr{
      \bigvee\setcompr{
        \bigwedge\setcompr{
          \bigwedge \alpha\impl a
        }{a\in\beta}
      }{\beta\in\brack t\cap\Cohf}
    }{\alpha\in\brack s\cap\Cohf}\]
  Now, we notice that for any $\alpha\in\Cohf$ and $a\in\alpha$, we have $\top\wedge\bigwedge\alpha\equiv_{\text{fan}}\bigwedge\alpha\leqq_{\text{fan}}a$, hence $\top\leqq_{\text{fan}}\bigwedge\alpha\impl a$.
  This means that we may derive: 
  \begin{align*}
   \bigwedge\setcompr{
    \bigwedge\alpha\impl a
    }{a\in\beta}
    &=\bigwedge\paren{\setcompr{
    \bigwedge\alpha\impl a
    }{a\in\beta\setminus\alpha}
      \cup
      \setcompr{
    \bigwedge\alpha\impl a
    }{a\in\beta\cap\alpha}}\\
    &\equiv_{\text{fan}}\bigwedge\paren{\setcompr{
    \bigwedge\alpha\impl a
    }{a\in\beta\setminus\alpha}
      \cup
      \setcompr{
    \top
      }{a\in\beta\cap\alpha}}\\
    &\equiv_{\text{fan}}\paren{\bigwedge\setcompr{
      \bigwedge\alpha\impl a
      }{a\in\beta\setminus\alpha}}
      \wedge
      \paren{\bigwedge
      \setcompr{
      \top
      }{a\in\beta\cap\alpha}}\\
    &\equiv_{\text{fan}}\paren{\bigwedge\setcompr{
      \bigwedge\alpha\impl a
      }{a\in\beta\setminus\alpha}}
      \wedge\top\\
    &\equiv_{\text{fan}}\bigwedge\setcompr{
      \bigwedge\alpha\impl a
      }{a\in\beta\setminus\alpha}.
  \end{align*}
  If we replace this in the expression for $s\impl t$, and further rewrite \eqref{ax:15} the identity becomes: 
  \[s\impl t\equiv_{\text{fan}}
    \bigwedge\setcompr{
      \bigvee\setcompr{
        \bigwedge\setcompr{
          \paren{\bigwedge \alpha\impl \bot}\vee a
        }{a\in\beta\setminus \alpha}
      }{\beta\in\brack t\cap\Cohf}
    }{\alpha\in\brack s\cap\Cohf}\]
  Finally, we use \eqref{ax:16} to get a $\Tlat$-term $s\impl t\equiv_{\text{fan}}$
  \[
    \bigwedge\setcompr{
      \bigvee\setcompr{
        \bigwedge\setcompr{
          \paren{\bigvee\setcompr{b}{\exists c\in\alpha:\,b\incoh c}}\vee a
        }{a\in\beta\setminus \alpha}
      }{\beta\in\brack t\cap\Cohf}
    }{\alpha\in\brack s\cap\Cohf}.\]
    \qedhere
  % We may simplify this term slightly using distributivity (this relies on the assumption that $\brack t\cap\Cohf\neq\emptyset$):
  % \begin{align*}
  %   s\impl t
  %   &\equiv_{\text{fan}}
  %     \bigwedge\setcompr{
  %     \bigvee\setcompr{
  %     \paren{\bigwedge \alpha\impl \bot}\vee
  %     \bigwedge\paren{\beta\setminus \alpha}
  %     }{\beta\in\brack t\cap\Cohf}
  %     }{\alpha\in\brack s\cap\Cohf}\\
  %   &\equiv_{\text{fan}}
  %     \bigwedge\setcompr{
  %     \paren{\bigwedge \alpha\impl \bot}\vee
  %     \bigvee\setcompr{
  %     \bigwedge\paren{\beta\setminus \alpha}
  %     }{\beta\in\brack t\cap\Cohf}
  %     }{\alpha\in\brack s\cap\Cohf}
  % \end{align*}
  % Finally, we use \eqref{ax:14} to get a $\Tlat$-term:
  % \[s\impl t\equiv_{\text{fan}}
  %     \bigwedge\setcompr{
  %     \paren{\bigvee\setcompr{b}{\exists c\in\alpha:\,b\incoh c}}\vee
  %     \bigvee\setcompr{
  %     \bigwedge\paren{\beta\setminus \alpha}
  %     }{\beta\in\brack t\cap\Cohf}
  %     }{\alpha\in\brack s\cap\Cohf}\qedhere\]
  
\end{proof}
\NewDocumentCommand\nf{d()}{
  \IfValueTF{#1}{%
    \mathrm{nf}{\paren{#1}}
  }{%
    \mathrm{nf}
  }
}
  
\begin{corollary}\label{cor:fan-repr}
  There is a representation $\tuple{\Tlat,\tau,id_{\Tlat}}$ of $G$ with respect to $\equiv_{\text{fan}}$. 
\end{corollary}
\begin{proof}
  $\tau$ is defined by induction on $s$. Every case is trivial, with the implication being handled by \Cref{lem:impl-to-lat} as $\tau(s\impl t)\eqdef I^{\tau(t)}_{\tau(s)}$. 
\end{proof}

\begin{theorem}\label{thm:obs-alg-compl}
  For any pair $s,t\in\Tobs$, we have $s\leqq_{\text{fan}}t$ if and only if $\sem s\subseteq\sem t$.
\end{theorem}
\begin{proof}
 By \Cref{lem:repr-impl-compl} and \Cref{cor:fan-repr}.
\end{proof}

\begin{lemma}
  The problem of testing whether $s\leqq_{\text{fan}}t$ is decidable.
\end{lemma}
\begin{proof}
  As remarked earlier, it is enough to compare $\sem{\tau(s)}$ and $\sem{\tau(t)}$.
  Using the notations of \Cref{sec:obs-lat}, we can see that $\sem{\tau(s)}=\paren{\brack{\tau(s)}\cap\Cohf}\closure$, with $\brack{\tau(s)}\cap\Cohf$ being a finite set of finite cliques.
  That set is easily computable, since $\tau(s)$ itself is computable (see the proof of \Cref{lem:impl-to-lat}).
  So, $\sem{\tau(s)}\subseteq\sem{\tau(t)}$ is equivalent to the following statement:
  \[\forall \alpha \in \brack{\tau(s)}\cap\Cohf,\,\exists \beta\in\brack{\tau(t)}\cap\Cohf:\;\beta\subseteq\alpha.\]
  Since we only quantify over finite sets, and since we only test containment of finite sets, this is decidable.
\end{proof}
\begin{remark}
  The decidability result we have just established does not mention complexity.
  This would not be to our advantage, since the computation of $\tau(s)$ is hopelessly costly: not only does is involve $\brack{s}$, which is already a large object, but implication steps generate exponential blow-ups (so nested implications yield towers of exponentials).
  It is possible this is sub-optimal, although we doubt a general algorithm of tractable complexity exists.
  Any practical application should thus focus on more restricted classes of graphs, that should be expressive enough for the use cases, and restrictive enough so efficent algorithms exist.
\end{remark}
\subsection{Infinite anticliques}
\label{sec:infinite-anticliques}

We now consider an infinite\footnote{Recall that finite graphs have the FAN property, and are thus covered by the previous sub-section.} set $\omega$, and define a coherence graph $\Omega\eqdef\tuple{\omega,id_\omega}$.
\begin{tableequations}
  We notice right away that cliques in this graph are either empty or singleton sets.
  Furthermore, we observe that for any element $a\in\omega$, the negation of $a$ is the set of all singleton cliques whose element is not $a$, i.e. we have:
  \[\sem{a\impl\bot}=\setcompr{\set b}{b\neq a}.\]
  As such, the term $a\vee\paren{a\impl\bot}$ represents the set of all non-empty cliques.
  Since this set does not depend on $a$, this leads to the following axiom:
  \begin{align}
    \forall a,b\in\omega,&&a\vee\paren{a\impl\bot}\equiv b\vee\paren{b\impl\bot}\axlabel{ax:17}
  \end{align}

  In Heyting algebras, the double negation is not in general the identity.
  Observation algebras are no exception, but in this graph the correspondence does hold for sets of singleton cliques as we will now describe.

  This holds trivially for the empty set (the negation exchanges $\top$ and $\bot$ as one would expect).
  If on the other hand $A\in\pset \omega$ is a non-empty set of elements, the set $x\eqdef\setcompr{\set a}{a\in A}$ is an observation (it is trivially down-closed, since $\alpha\preceq\set a$ entails $\alpha=\set a$).
  We may thus perform the following computation:
  \begin{align*}
    \alpha\in x\impl\emptyset
    &\Leftrightarrow \forall \beta\in x,\,\alpha\cup\beta\notin\Coh(\Omega)\\
    &\Leftrightarrow \forall \beta\in x,\,\exists a\neq b:\alpha=\set a\wedge \beta=\set b\\
    &\Leftrightarrow \forall b\in A,\,\exists a\neq b:\alpha=\set a\\
    &\Leftrightarrow \exists b,\,\alpha=\set b\wedge b\notin A
  \end{align*}
  Note the last equivalence relies on $A$ being non-empty.
  This means that $x\impl\emptyset=\setcompr{\set a}{a\notin A}$, i.e. this observation is isomorphic to the complement of $A$ with respect to $\omega$.
  If $A$ is chosen to be non-trivial (non-empty nor equal to $\omega$), then its negation is itself non-empty, and we may use the same computation a second time to get $\paren{x\impl\emptyset}\impl\emptyset=x$.
  In particular, since $\omega$ was assumed infinite, for any finite set $A$ we get the following axiom:
  \begin{align}
    \forall A\in\fpset \omega,&&\paren{\paren{\bigvee A}\impl\bot}\impl \bot\equiv\bigvee A\axlabel{ax:18}
  \end{align}
\end{tableequations}
This defines the relation $\equiv_{\omega}$, as the least congruence satisfying  \eqref{ax:1}-\eqref{ax:8}, \eqref{ax:13}, \eqref{ax:9}-\eqref{ax:12}, \eqref{ax:14}-\eqref{ax:15}, as well as axioms~\eqref{ax:17} and \eqref{ax:18}.
As before, soundness of this relation is trivial from our results so far.

\newcommand\Top{\mathrm{Top}}
\NewDocumentCommand\Fin{ d() }{\IfValueTF{#1}{\mathrm{Fin}{\paren{#1}}}{\mathrm{Fin}}}
\NewDocumentCommand\CoFin{ d() }{\IfValueTF{#1}{\mathrm{CoFin}{\paren{#1}}}{\mathrm{CoFin}}}
We now set out to define a representation triple for $\Omega$ with respect to $\equiv_\omega$.
Our set of representatives $R$ can be described as $1+2\times\fpset \omega$, or equivalently as a datatype with three constructors: a constant $\Top:R$, as well as $\Fin,\CoFin:\fpset \omega\to R$.
To give intuitive meaning to these, we define their interpretations as follows:
\begin{align*}
  \sem{\Top}_R&\eqdef\Coh(\Omega)
  &\sem{\Fin(A)}_R&\eqdef\setcompr{\set a}{a\in A}
  &\sem{\CoFin(A)}_R&\eqdef\setcompr{\set a}{a\notin A}
\end{align*}
$\Top$ is interpreted as the full observation, $\Fin(A)$ as the finite set of singleton cliques whose elements are in $A$, and $\CoFin(A)$ as the co-finite set of singleton cliques whose elements are not in $A$.

This intuition leads to the following translation from $R$ to $\Tobs$:
\begin{align*}
  \phi(\Top)&\eqdef \top
  &\phi(\Fin(A))&\eqdef\bigvee A
  &\phi(\CoFin(A))&\eqdef\begin{cases}(\bigvee A)\to\bot&if~A\neq\emptyset\\a\vee(a\to\bot)&if~A=\emptyset
  \end{cases}
\end{align*}
where $a\in\omega$ is some arbitrary atomic observation. Thanks to \eqref{ax:17}, the choice of $a$ does not matter (and $\omega$ being infinite, it is not prepostrous to imagine one could find such an element).
It is straighforward to check that our ``intuitive'' interpretation is consistent with this translation, in the sense that $\sem{\phi(r)}=\sem r_R$.
Proving that $\sem{\phi(r)}=\sem{\phi(r')}$ entails $\phi(r)\equiv_\omega\phi(r')$ is also very simple thanks to the following remarks:
\begin{enumerate}
\item\label{item:1} $\emptyset \in\sem{r}_R$ if and only if $r=\Top$;
\item\label{item:2} for any pair of finite sets $A,B$, since $\sem{\Fin(A)}_R$ is finite and $\sem{\CoFin(B)}_R$ is not, we have
$$\sem{\Fin(A)}_R\neq\sem{\CoFin(B)}_R;$$
\item\label{item:3} for any finite set $A$, we have $A=\setcompr{a}{\set a\in\sem{\Fin(A)}_R}=\setcompr{a}{\set a\notin\sem{\CoFin(A)}_R}$.
\end{enumerate}
Together, these remark entail that $\sem-_R$ is injective, thus $\sem{\phi(r)}=\sem{\phi(r')}$ entails $\sem r_R=\sem{\phi(r)}=\sem{\phi(r')}=\sem{r'}_R$, thus $r=r'$.
By reflexivity, this means $\phi(r)\equiv_\omega\phi(r')$.
\begin{remark}
  Since $\sem{}_R$ is injective and  $\sem{}\circ\phi=\sem{}_R$, the question of testing given $r$ and $r'$ whether $\sem{\phi(r)}$ is equal to $\sem{\phi(r')}$ is decidable by simply comparing $r$ and $r'$.
  Provided an appropriate encoding of sets, this can be done efficiently (in linear time if elements are represented as integers~\cite{KATRIEL2004175}).
\end{remark}

So, what remains to be considered is the translation from $\Tobs$ terms to $R$, which we shall build by induction.
The constants and atomic cases are straightforward:
\begin{align*}
  \tau(\top)&\eqdef \Top
  &\tau(\bot)&\eqdef \Fin(\emptyset)
  &\tau(a)&\eqdef \Fin(\set a).
\end{align*}
Disjuntions and conjunctions are also relatively easy to handle:
\begin{align*}
  \tau(s\vee t)&\eqdef \tau(s)\oplus \tau(t)
  &\tau(s\wedge t)&\eqdef \tau(s)\otimes \tau(t)\\
  \Top\oplus r=r\oplus\Top&\eqdef \Top
  &\Top\otimes r=r\otimes\Top&\eqdef r\\
  \Fin(A)\oplus \Fin(B)&\eqdef \Fin(A\cup B)
  &\Fin(A)\otimes \Fin(B)&\eqdef \Fin(A\cap B)\\
  \CoFin(A)\oplus \CoFin(B)&\eqdef \CoFin(A\cap B)
  &\CoFin(A)\otimes \CoFin(B)&\eqdef \CoFin(A\cup B)\\
  \Fin(A)\oplus \CoFin(B)&\eqdef \CoFin(B\setminus A)
  &\Fin(A)\otimes \CoFin(B)&\eqdef \Fin(A\setminus B)\\
  \CoFin(A)\oplus \Fin(B)&\eqdef \CoFin(A\setminus B)
  &\CoFin(A)\otimes \Fin(B)&\eqdef \Fin(B\setminus A)
\end{align*}

We dispatch the only remaing case, that of the implication,  by setting 
$\tau(s\impl t)\eqdef \tau(s)\ominus \tau(t)$ and defining:
\begin{align*}
  r\ominus \Top&\eqdef \Top
  &
  \Top\ominus r&\eqdef r\\
  \Fin(A)\ominus \Fin(B)&\eqdef
                          \begin{cases}
                            \Top&\text{if}~A\subseteq B\\
                            \Fin(A\setminus B)&\text{otherwise}
                          \end{cases}
  \\
  \Fin(A)\ominus \CoFin(B)&\eqdef 
                              \begin{cases}
                                \Top&\text{if}~A\cap B=\emptyset\\
                                \CoFin(B\cap A)&\text{otherwise}
                              \end{cases}
   \\
  \CoFin(A)\ominus \CoFin(B)&\eqdef
                              \begin{cases}
                                \Top&\text{if}~B\subseteq A\\
                                \CoFin(B\setminus A)&\text{otherwise}
                              \end{cases}\\
  \CoFin(A)\ominus \Fin(B)&\eqdef \Fin(A\cup B)
\end{align*}
Note that the case analyses in this definition are necessary to ensure that whenever $\sem r_R\subseteq\sem{r'}_R$ we have $r\ominus r'=\Top$.

% It is not difficult (if rather tedious) to check that these definitions is semantically correct, i.e.
% \begin{align*}
% \sem{r\oplus r'}_R&=\sem{r}_R\cup\sem{r'}_R&
% \sem{r\otimes r'}_R&=\sem{r}_R\cap\sem{r'}_R&
% \sem{r\ominus r'}_R&=\sem{r}_R\impl\sem{r'}_R.
                       %   \end{align*}
Now we need to check that $\phi\circ\tau(s)\equiv_\omega s$.
This is done by induction on $s$, the constant and atomic cases being trivial.
The key steps consist in checking that for arbitrary representations $r,r'\in R$ we have: 
\begin{align*}
\phi\paren{r\oplus r'}&\equiv_\omega\phi\paren{r}\vee\phi\paren{r'}&
\phi\paren{r\otimes r'}&\equiv_\omega\phi\paren{r}\wedge\phi\paren{r'}&
\phi\paren{r\ominus r'}&\equiv_\omega\phi\paren{r}\impl\phi\paren{r'}.
\end{align*}
This proof is a (lenghty) case analysis.
It has been verified in Rocq, and as such we omit it.

\begin{theorem}\label{thm:compl-ac}
  For any pair of terms $s,t\in\Tobs(\Omega)$, we have $s\leqq_\omega t$ if and only if $\sem s\subseteq\sem t$.
\end{theorem}
\begin{proof}
  By \Cref{lem:repr-impl-compl}.
\end{proof}

\begin{lemma}
  The problem of testing whether $s\equiv_\omega t$ is decidable.
\end{lemma}
\begin{proof}
  This is equivalent to checking if $\sem{s}=\sem t$, i.e. $\sem{\phi\circ\tau(s)}=\sem{\phi\circ\tau(t)}$.
  As noticed earlier, this amount to the equality of $\tau(s)$ and $\tau(t)$, which is decidable.
\end{proof}
\begin{remark}
  The complexity of this problem is somewhat better that in the previous case.
  Indeed, checking whether $\tau(s)=\tau(t)$ can be done efficiently with respect to the sizes of $\tau(s)$ and $\tau(t)$.
  These sizes are themselves linear in the size of $s$ and $t$.
  The bottleneck for complexity is thus the computation of $\tau$ itself.
  From the definition, we gather that this computation involves a linear number of set operations (inclusion testing, union, intersection, differences), each of the sets involved being of linear size (only atomic observations appearing explicitely in the term may be used).
  So, this remains polynomial, with a cubic rough upper-bound.
\end{remark}

\clearpage
\section{Product observation algebras}
\label{sec:prods}
\def\G{ G}
\def\at{\mathbin{@}}
\def\si{s_{{\paren i}}}
\def\ti{t_{{\paren i}}}
\def\I{\mathcal I}
\newcommand\support[1]{\left|#1\right|}

In this section, we present a product construction, building a coherence graph, and an axiomatisation of the corresponding observation algebra, out of smaller graphs.

Let $\I$ be an arbitrary set of dimensions, and for each $i\in\I$ let $G_i=\tuple{\atobs_i,\coh_i}$ be some coherence graph.
We also assume for each $i\in\I$ a sound equivalence relation $\equiv_i$ over $\Tobs(\G_i)$, meaning in particular that
\[\forall i\in \I,\,\forall \si,\ti\in\Tobs(\G_i),\; \si\leqq_i\ti\Rightarrow\sem \si\subseteq\sem\ti.\]
\begin{definition}
The product of the family $\tuple{G_i}_{i\in \I}$ is the graph $\bigotimes_i\G_i\eqdef\tuple{\atobs,\coh}$ where:
\begin{mathpar}
  \atobs\eqdef\setcompr{a\at i}{i\in \I,\,a\in \atobs_i}\and
  \inferrule{i\neq j}{a\at i\coh b\at j}\and
  \inferrule{a\coh_i b}{a\at i\coh b\at i}
\end{mathpar}
\end{definition}

This construction preserves the FAN property, in the following sense:
\begin{lemma}\label{lem:prod-fan}
  If for every $i\in\I$ the graph $\G_i$ has the FAN property, so does $\bigotimes_{i\in\I}\G_i$.
\end{lemma}
\begin{proof}
  Let $a\in O$. By definition there is $i\in \I$ and $a'\in O_i$ such that $a=a'\at i$. Therefore: 
  \[\setcompr{b\in O}{b\incoh a}=\setcompr{b\at i}{b\incoh_ia'}\]
  The anti-neighbourhood of $a$ in $\bigotimes_i\G_i$ is thus isomorphic to that of $a'$ in $\G_i$, which is finite. 
\end{proof}

After reviewing a few examples in \Cref{sec:examples}, we will propose an axiomatisation of the observation algebra corresponding to this product graph in \Cref{sec:axioms-term-vectors}.
This axiomatisation will rely on a notion of term vectors, which will then be used in \Cref{sec:repr-compl} to define a pair of representation triples, thus proving the completeness of our axiomatisation.

To simplify notations in \Cref{sec:axioms-term-vectors,sec:repr-compl}, we will write $\P\eqdef\bigotimes_i\G_i$.

\subsection{Examples}
\label{sec:examples}
\def\eqeq{\mathrel{{=}{=}}}

In the binary case, when $\I=\set{1,2}$, this corresponds to the graph over the disjoint union of sets of vertices, and where the coherence relation is $\coh_1\cup\coh_2\cup\paren{O_1\times O_2}\cup\paren{O_2\times O_1}$.

We denote by $\mathcal B$ the binary graph $\tuple{\set{0,1},\emptyset}$, with two incoherent atomic observations.
The graph $\bigotimes_{i\in\Nat}\mathcal B$ may then be understood as representing an infinite memory, containing $\Nat$-indexed boolean cells.
Its atomic observations may be denoted $v{\brack i}\eqeq b$, with $i\in\Nat$, $b\in\set{0,1}$, and $v{\brack i}$ denoting the $i^{\text{th}}$ cell in the memory.
We may enrich the specification language (a.k.a. term syntax) with statements $v{\brack i}\eqeq v{\brack j}$.
Indeed, such a statement may be encoded as $\paren{v{\brack i}\eqeq 0 \wedge v{\brack j}\eqeq 0}\vee\paren{v{\brack i}\eqeq 1\wedge v{\brack j}\eqeq 1}$.

\def\V{\mathcal V}
\def\M{\mathcal M}
More generally, given a collection $\paren{\V_i}_{i\in\Nat}$ of finite sets of values, we can define the graph $\M\eqdef\bigotimes_{i\in\Nat}\tuple{\V_i,\emptyset}$.
We can also define the following terms for $i,j\in\Nat$ and $a\in\V_i$:
\begin{align*}
  \paren{v{\brack i}\eqeq a}&\eqdef a\at i
  &\paren{v{\brack i}\eqeq v{\brack j}}&\eqdef\bigvee_{a\in\V_i\cap\V_j}\paren{v{\brack i}\eqeq a\wedge v{\brack j}\eqeq a}
\end{align*}
This graph represents a memory with infinity many cells, each of which containing a value chosen from a finite set.
Note that in this model different cells may have different sets of possible values, i.e. a simple form of variable typing is available.
Another important remark is that since each of the $\V_i$ is finite, the graphs $\tuple{\V_i,\emptyset}$ all have the FAN property.
By \Cref{lem:prod-fan} $\M$ inherits this property as well.

To move away from the FAN-case, we need to use non-FAN graphs as $\G_i$.
For instance, consider the anticlique described in~\Cref{sec:infinite-anticliques} with $\omega\eqdef\Nat$ as atomic observations.
Taking an infinite product of such structures leads to a graph describing an infinite set of natural number valued variables.
As in the previous case, the atomic observations should be understood as statements $v_i\eqeq n$.
However, it is no longer possible to encode $v_i\eqeq v_j$, as this would involve an infinite disjunction between all possible values.
In other words, addding $v_i\eqeq v_j$ to our syntax modifies the set of expressible observations.
A finer study of this construct would be an interesting path for future research.

\subsection{Axioms and term vectors}
\label{sec:axioms-term-vectors}
\begin{tableequations}

  The first thing we observe on this algebra is that it contains each of the $\Obs(\G_i)$ as sub-algebras.
  %% TODO: clarify that we only consider the representable part of the algebra.
  Indeed, for any index $i\in \I$, and any term in $\si\in\Tobs(\G_i)$, we may build a term $\underline\si\in\Tobs(\P)$ by replacing any occurence of an atomic observation $a\in\atobs_i$ by $a\at i\in\atobs$.
  The semantics of this term is as follows:
  \begin{equation}
    \alpha\in\sem{\underline\si}\Leftrightarrow\setcompr{a\in\atobs_i}{a\at i\in\alpha}\in\sem \si.\label{eq:def-inject}
  \end{equation}
  From this statement, whose proof is a straight forward induction on terms, we obtain that this transformation is injective with respect to the semantics.
  This motivates the following axiom (scheme):
  \begin{equation}
    \forall \si,\ti\in\Tobs(\G_i),\;\si\equiv_i\ti \Rightarrow \underline\si\equiv \underline\ti\axlabel{ax:19}
  \end{equation}

  Also observe that the cliques in $\P$ are unions of cliques from the $G_i$s.
  As such they may be represented as $\I$-indexed vectors, where the $i^{\text{th}}$ component belongs to $\Coh(G_i)$.
  Furthermore, finite cliques are such vectors $v$ where there is a finite subset $I\subseteq \I$ such that $i\notin I\Rightarrow v_i=\emptyset$.
  The least such $I$ is called the support of $v$.

  This remark will motivate us to see terms from $\Tobs(\P)$ as lattice terms over $\Tobs(\G_i)$.
  To define this representation, we rely on the notion of term vectors.
  \begin{definition}
    A term vector $v$ is a partial function whose domain (also called support) is a finite set $\support v\subseteq\I$.
    For each $i\in\support v$, the $i^{\text{th}}$ coordinate of $v$ is a term $v_i\in\Tobs(\G_i)$.
    We write $\vec{\Tobs}$ for the set of term vectors, $\vec 0$ for the term vector with empty support, and define the following pair of maps $\vec\Tobs\to\Tobs(\P)$:
    \begin{align*}
      \prod v&\eqdef\bigwedge_{i\in\support v}\underline{v_i}
      &\coprod v&\eqdef\bigvee_{i\in\support v}\underline{v_i}
    \end{align*}
    We also define the conjunction, disjunction, and implication of a pair of term vectors $u,v$, all of which have the support $\support u\cup\support v$:
    \begin{align*}
      (u\vee v)_i&\eqdef
                   \begin{cases}
                     u_i\vee v_i&i\in\support u\cap\support v\\
                     v_i&i\in\support v\setminus\support u\\
                     u_i&i\in\support u\setminus\support v
                   \end{cases}
                                &
                                  (u\wedge v)_i&\eqdef
                                                 \begin{cases}
                                                   u_i\wedge v_i&i\in\support u\cap\support v\\
                                                   v_i&i\in\support v\setminus\support u\\
                                                   u_i&i\in\support u\setminus\support v
                                                 \end{cases}
                                \\
                                  (u\to v)_i&\eqdef
                                              \begin{cases}
                                                u_i\to v_i&i\in\support u\cap\support v\\
                                                v_i&i\in\support v\setminus\support u\\
                                                u_i\to\bot&i\in\support u\setminus\support v
                                              \end{cases}
    \end{align*}
  \end{definition}
  \noindent%
  Note that using the axioms of distributive lattices, we have $$\prod(u\wedge v)\equiv \prod u\wedge\prod v\text{ and }\coprod(u\vee v)\equiv \coprod u\vee\coprod v.$$
  
  With these notations, we may state our last axiom:
  \begin{equation}
    % \support u=\support v\Rightarrow
    \paren{\prod u}\to\paren{\coprod v}\equiv \coprod\paren{u\to v}
    \axlabel{ax:20}
  \end{equation}
  \begin{lemma}
    \eqref{ax:20} is sound.
  \end{lemma}
  \begin{proof}
    We start by assuming $\support u = \support v$.
    (If this is not the case, we pad $u$ with $\top$ and $v$ with $\bot$ so that their support coincide.
    This transformation yields $u'$ and $v'$ such that $\prod u\equiv\prod u'$, $\coprod v\equiv\coprod v'$, and $\coprod u\to v\equiv \coprod u'\to v'$.)
    
    One direction of this identity holds trivially: since we have $\sem{\prod \underline{u_i}}\subseteq\sem{\underline {u_i}}$ and $\sem{\underline{v_i}}\subseteq\sem{\coprod v}$, we have $\sem{\underline{u_i}\to\underline{v_i}}\subseteq\sem{\prod u\to\coprod v}$, hence
    \begin{align*}
      \sem{\coprod\paren{u\to v}}=\sem{\bigvee_i\paren{\underline{u_i}\to\underline{v_i}}}=\bigcup_i\sem{\underline{u_i}\to\underline{v_i}}\subseteq\sem{\prod u\to\coprod v}
    \end{align*}

    For the converse containment, let $\alpha\in\sem{\prod u\to\coprod v}$.
    As noticed earlier, $\alpha$ may be seen as a vector indexed over $\I$, and whose $i^{\text{th}}$ component is a clique $\alpha_i=\setcompr{a\in\atobs_i}{a\at i\in\alpha}\in\Coh(\G_i)$.
    The fact that $\alpha\in\sem{\prod u\to\coprod v}$ means that for any $\beta\in\sem{\prod u}$, either $\alpha\cup\beta\notin\Coh(\P)$, or $\alpha\cup\beta\in\sem{\coprod v}$.
    To prove that $\alpha$ belong to the semantics of the right-hand side, we must find some index $i$ such that $\alpha_i\in\sem{u_i\to v_i}$.
    By contradiction, assume no such $i$ exists, that is:
    \[
      \forall i\in\support u,\,\exists \beta_i\in\sem{u_i}:\alpha_i\cup\beta\in\Coh(\G_i)\text{ and }\alpha_i\cup\beta\notin\sem{v_i}
    \]
    We define $\beta\eqdef\setcompr{a\at i}{i\in\support u,\,a\in\beta_i}$.
    We may check that for any $i\in\support u$, $\beta\in\sem{\underline{u_i}}$ (by~\eqref{eq:def-inject}), meaning $\beta\in\sem{\prod u}$.
    Also, $\beta\cup\alpha\in\Coh(\P)$, since for each $i\in\support u$ we have $\beta_i\cup\alpha_i\in\Coh(\G_i)$, and for each $j\notin\support u$, we have $\beta_j=\emptyset$.
    So, since $\alpha\in\sem{\prod u\to\coprod v}$, we should have $\alpha\cup\beta\in\sem{\coprod v}$, hence some index $i\in\support u$ such that:
    \[\alpha_i\cup\beta_i\in\sem{v_i}.\]
    This is however a contradiction with our assumption, thus concluding the proof.
  \end{proof}
\end{tableequations}
We may now define an axiomatisation $\equiv_\P$ of the observation algebra over $\P$, by taking the axioms \eqref{ax:1}-\eqref{ax:8}, \eqref{ax:13}, and \eqref{ax:9}-\eqref{ax:12}, as well as \eqref{ax:19} and \eqref{ax:20}.
This relation is clearly sound, as each axiom has been proved sound already.

\begin{remark}[On the finiteness of axiomatisations]
  The axiomatisation defined here is always infinite, even if the set of dimensions is finite and if each $\G_i$ admits a finite axiomatisation.
  Indeed, the axiom scheme~\eqref{ax:19} will introduce an axiom for each identity in $\G_i$.
  The naive solution, which would just collect the axioms defining each $\equiv_i$ would be incorrect in general: universal axioms may hold on one of the $\G_i$, but not in the product algebra.
  
  One might devise a system with typed terms that would allow one to overcome this difficulty, at least in some cases.
  It is not clear however how useful such a system would be.
  Indeed, if the graph is finite, thanks to the development in~\Cref{sec:fan-graphs}, we know of a complete axiomatisation.
  If it is infinite, since we need to take care of all of the incoherent pairs of atomic observations, it is doubtful any finite axiomatisation exists.
  For these reasons, we are satisfied with the relation $\equiv_\P$ as defined here, for it is fairly simple to describe and manipulate.
\end{remark}
%% TODO : Discuss the case when the \equiv_i are generated by axioms  \eqref{ax:1}-\eqref{ax:8}, \eqref{ax:13}, and \eqref{ax:9}-\eqref{ax:12}, plus some closed axioms.

\subsection{Representations and completeness}
\label{sec:repr-compl}

We will now prove that this relation is also complete for $\Obs(\P)$, provided each of the $\equiv_i$ relations is complete, meaning:
\[\forall i\in \I,\,\forall \si,\ti\in\Tobs(\G_i),\;\sem \si\subseteq\sem\ti\Rightarrow \si\leqq_i\ti.\]

As usual, we will use representation triples.
Here it is useful to use two representations triples, a disjunctive and a conjuntive one.
Both use the same set of representatives, namely finite sets of term vectors.
\[R\eqdef \fpset{\vec\Tobs}\]
The disjunctive representation has a translation $\phi_\vee(V)\eqdef\bigvee_{v\in V}\prod v$, meaning it is understood as a disjunction of term vectors.
Similarly, the conjunctive representation holds a translation $\phi_\wedge(V)\eqdef\bigwedge_{v\in V}\coprod v$.
Observe that we may switch from one representation to the other, using the following helper functions:
\begin{align*}
  C2D(\emptyset)&=\set{\vec 0}\\
  C2D(\set v\uplus V)&=\setcompr{\vec{v_i}\wedge u}{i\in\support v,\,u\in C2D(V)}\\
  D2C(\emptyset)&=\set{\vec 0}\\
  D2C(\set v\uplus V)&=\setcompr{\vec{v_i}\vee u}{i\in\support v,\,u\in D2C(V)}.
\end{align*}
(Where $\vec{v_i}$ represents the vector with domain $\set i$ whose $i^{\text{th}}$ component is $v_i$.)

These functions are best understood by the following statement:
\begin{lemma}\label{lem:c2d-d2c}
  For any representative $V\in R$, we have:
  \begin{align*}
    \phi_\vee\paren{C2D(V)}&\equiv_\P\phi_\wedge(V)
    &\phi_\wedge\paren{D2C(V)}&\equiv_\P\phi_\vee(V)
  \end{align*}
\end{lemma}
\begin{proof}
  By induction on the cardinal of $V$:
  \begin{align*}
    \phi_\vee\paren{C2D(\emptyset)}
    &= \phi_\vee\paren{\set{\vec 0}}
      =\prod \vec 0
      =\top
      =\bigwedge\emptyset=\phi_\wedge(\emptyset)=\phi_\wedge(V)\\
    \phi_\vee\paren{C2D(\set v\uplus V)}
    &= \phi_\vee\paren{\setcompr{\vec{v_i}\wedge u}{i\in\support v,\,u\in C2D(V)}}
    = \bigvee_{i\in\support v,\,u\in C2D(V)}\prod (\vec{v_i}\wedge u)\\
    &\equiv_\P\bigvee_{i\in\support v,\,u\in C2D(V)}\paren{\prod \vec{v_i}}\wedge\paren{\prod u}
    =\bigvee_{i\in\support v,\,u\in C2D(V)}\underline{v_i}\wedge\paren{\prod u}\\
    &\equiv_\P\bigvee_{i\in\support v}\underline{v_i}\wedge\paren{\bigvee_{u\in C2D(V)}\prod u}
    =\coprod v\wedge\paren{\bigvee_{u\in C2D(V)}\prod u}\\
    &=\coprod v\wedge\phi_\vee\paren{C2D(V)}
      \equiv_\P\coprod v\wedge\phi_\wedge\paren{V}
    =\phi_\wedge(\set v\uplus V)
  \end{align*}
  The case of $\phi_\wedge\circ D2C$ is symmetric.\qedhere
\end{proof}

We now build the corresponding converse translations, $\tau_\vee$ and $\tau_\wedge$.
Thanks to the functions defined in \Cref{lem:c2d-d2c} we only need to define $\tau_\vee$ by induction on terms, as we may set $\tau_\wedge(s)=D2C(\tau_\vee(s))$.
\begin{itemize}
\item For atomic observations, $\tau_\vee(a\at i)$ is defined as the singleton containing the term vector with support $\set i$, whose $i^{\text{th}}$ component is $a$.
\item For the constants $\top,\bot$, we have $\tau_\vee(\top)\eqdef \set{\vec 0}$
 and $\tau_\vee(\bot)\eqdef \emptyset$.
\item Disjunctions are straightforward:
  $\tau_\vee(s\vee t)=\tau_\vee(s)\cup\tau_\vee(t)$.
\item For conjuctions, we rely on our helper functions:
  $$\tau_\vee(s\wedge t)=C2D\paren{D2C\paren{\tau_\vee(s)}\cup D2C\paren{\tau_\vee(t)}}.$$
\item The remaining case, implication, is where the axiom \eqref{ax:20} comes into play.
  % We inductively compute $U\eqdef \tau_\vee s$ and $V\eqdef\tau_\vee t$.
  % We normalise both, so that there is a finite set of indices $I\subseteq\I$ such that for every $u\in U\cup V$, $\support u=I$.
  % This is achieved by choosing $I$ as the union of the support of all vectors in either $U$ or $V$, and augmenting each vector $u\in U\cup V$ by setting $u_i\eqdef \top$ whenever $i\in I\setminus \support u$.
  % Such a transformation does not change the disjunctive interpretation of representatives, up-to $\equiv_\P$.
  We define
  \[\tau_\vee\paren{s\to t}\eqdef C2D\setcompr{u\to v}{u\in \tau_\vee (s),\,v\in D2C\circ\tau_\vee(t)}.\]
\end{itemize}

These translations being defined, we first prove their compatibility with our axiomatisation, i.e. the following statement for the disjunctive representation.
\begin{lemma}
  For any term $s\in\Tobs(\P)$, we have $\phi_\vee\circ\tau_\vee(s)\equiv_\P s$.
\end{lemma}
\begin{proof}
  This proof is by induction on terms.
  The cases of constants, atomic observations, and disjunctions being trivial, we omit them here.
  \begin{align*}
    \phi_\vee\circ\tau_\vee(s\wedge t)
    &=\phi_\vee\circ C2D\paren{D2C\paren{\tau_\vee(s)}\cup D2C\paren{\tau_\vee(t)}}\\
    &\equiv\phi_\wedge\paren{D2C\paren{\tau_\vee(s)}\cup D2C\paren{\tau_\vee(t)}}\\
    &\equiv\paren{\phi_\wedge\circ D2C\paren{\tau_\vee(s)}}
    \wedge \paren{\phi_\wedge\circ D2C\paren{\tau_\vee(t)}}\\
    &\equiv\paren{\phi_\vee\paren{\tau_\vee(s)}}
      \wedge \paren{\phi_\vee\paren{\tau_\vee(t)}}\equiv s \wedge t\\  
  \end{align*}
  \begin{align*}
    \phi_\vee\circ\tau_\vee(s\to t)
    &=\phi_\vee\circ C2D\setcompr{u\to v}{u\in \tau_\vee (s),\,v\in D2C\circ\tau_\vee(t)}\\
    &\equiv_\P\phi_\wedge\setcompr{u\to v}{u\in \tau_\vee (s),\,v\in D2C\circ\tau_\vee(t)}\\
    &\equiv_\P\bigwedge_{u\in \tau_\vee (s),\,v\in D2C\circ\tau_\vee(t)}\coprod\paren{u\to v}\\
    &\equiv_\P\bigwedge_{u\in \tau_\vee (s),\,v\in D2C\circ\tau_\vee(t)}\paren{\prod u}\to \paren{\coprod v}
    \equiv_\P  \bigwedge_{v\in D2C\circ\tau_\vee t}\paren{\bigwedge_{u\in \tau_\vee s}\paren{\prod u\to \coprod v}}\\
    &\equiv_\P  \bigwedge_{v\in D2C\circ\tau_\vee t}\paren{(\bigvee_{u\in \tau_\vee s}\prod u)\to \coprod v}\\
    &\equiv_\P  (\bigvee_{u\in \tau_\vee s}\prod u)\to (\bigwedge_{v\in D2C\circ\tau_\vee t}\coprod v)\\
    &\equiv_\P  (\phi_\vee\circ\tau_\vee s)\to (\phi_\wedge\circ D2C\circ\tau_\vee t)\\
    &\equiv_\P  (\phi_\vee\circ\tau_\vee s)\to (\phi_\vee\circ\tau_\vee t)
      \equiv_\P s\to t\qedhere
  \end{align*}
\end{proof}
This implies $\phi_\wedge\circ\tau_\wedge(s)=\phi_\wedge\circ D2C\circ\tau_\vee(s)\equiv\phi_\vee\circ\tau_\vee(s)\equiv s$.
To obtain the remaining property of representation triples, we prove the following:
\begin{lemma}\label{lem:tchn-prod}
  For any $U,V\in R$, the following entailment hold:
  \begin{align*}
    \sem{\phi_\vee(U)}\subseteq\sem{\phi_\wedge(V)}&\Rightarrow\phi_\vee(U)\leqq_\P\phi_\wedge(V).
  \end{align*}
\end{lemma}
\begin{proof}
  Assume $\sem{\phi_\vee(U)}\subseteq\sem{\phi_\wedge(V)}$.
  By unfolding the definitions, we get:
  \[\bigcup_{u\in U}\sem{\prod u}
    =\sem{\phi_\vee(U)}
    \subseteq\sem{\phi_\wedge(V)}
    =\bigcap_{v\in V}\sem{\coprod v}.\]
  Therefore for every $u\in U$ and every $v\in V$, we have $\sem {\prod u}\subseteq\sem{\coprod v}$.
  By properties of Heyting algebras, this is equivalent to saying that $\Coh\subseteq\sem {\prod u\to\coprod v}$.
  By \eqref{ax:20}, we may rewrite this as $\Coh\subseteq\sem {\coprod \paren{u\to v}}$.
  Unfolding the definitions further, we obtain that for every\footnote{We omit the case where $\support u\neq\support v$ for brevity. This case poses no problem, and is left as an exercise.} $i$ we have 
  $\Coh\subseteq\sem{\underline{u_i\to v_i}}$.
  Using \eqref{eq:def-inject}, we can see that this entails $\sem\top_{\G_i}=\Coh(\G_i)\subseteq\sem{u_i\to v_i}$.
  By completeness of $\equiv_i$, we further obtain that for every $u\in U$, every $v\in V$, and every index $i$, we have $\top\leqq_iu_i\to v_i$, hence $u_i\leqq_i v_i$.
  Using \eqref{ax:19} we infer that $\underline {u_i}\leqq_\P \underline{v_i}$.
  We may now conclude:
  \begin{align*}
    &\forall u \in U,\forall v \in V,\forall i\in I:
      \underline {u_i}\leqq_\P \underline{v_i}\\
    \Rightarrow~
    &\forall u \in U,\forall v \in V,\forall i\in I:
      \underline {u_i}\leqq_\P \coprod v\\
    \Rightarrow~
    &\forall u \in U,\forall v \in V:
      \prod u\leqq_\P \coprod v\\
    \Leftrightarrow~
    &\forall u \in U:
      \prod u\leqq_\P \bigwedge_{v\in V}\coprod v
    \\
    \Leftrightarrow~
    & \phi_\vee(U)=\bigvee_{u\in U}\prod u\leqq_\P \bigwedge_{v\in V}\coprod v   = \phi_\wedge(V)
      \qedhere
  \end{align*} 
\end{proof}
  
As a corollary, we get that both our disjunctive and conjunctive frameworks are correct representations.
\begin{corollary}\label{cor:prod-repr}
  Both $\tuple{R,\tau_\vee,\phi_\vee}$ and $\tuple{R,\tau_\wedge,\phi_\wedge}$ are representation triples.
\end{corollary}
\begin{proof}
  Assume $\sem{\phi_\vee(U)}=\sem{\phi_\vee(V)}$.
  In particular this means $\sem{\phi_\vee(U)}\subseteq\sem{\phi_\vee(V)}$
  and $\sem{\phi_\vee(V)}\subseteq\sem{\phi_\vee(U)}$.
  Using our \Cref{lem:c2d-d2c} functions, this yields:
  $\sem{\phi_\vee(U)}\subseteq\sem{\phi_\wedge\circ C2D(V)}$
  and $\sem{\phi_\vee(V)}\subseteq\sem{\phi_\wedge\circ C2D(U)}$.
  By \Cref{lem:tchn-prod}, this means $\phi_\vee(U)\leqq_\P\phi_\wedge\circ C2D(V)\equiv_\P\phi_\vee(V)$
  and $\phi_\vee(V)\leqq_\P\phi_\wedge\circ C2D(U)\equiv_\P\phi_\vee(U)$.
  The proof for the other triple is similar.
\end{proof}
\begin{theorem}\label{thm:prod-alg-compl}
  For any pair $s,t\in\Tobs(\P)$, we have $s\leqq_\P t$ if and only if $\sem s\subseteq\sem t$.
\end{theorem}
\begin{proof}
 By \Cref{lem:repr-impl-compl} and \Cref{cor:prod-repr}.
\end{proof}
\begin{remark}
  In this case we may obtain decidability whenever each of the $\Obs(\G_i)$ is decidable.
  Indeed, in the proof of Lemma~\ref{lem:tchn-prod}, it can be seen that it is possible to reduce the containment of a pair of terms over $\P$ to a series of containment queries over the $\G_i$s.
  This reduction is however quite costly, with the computation of $\tau_\vee$ yielding a tower of exponentials.
  A finer study of the complexity of this problem could be the topic of further research.
\end{remark}

\clearpage
\section{Futher remarks}\label{sec:remarks}
\subsection{On the Rocq formalisation}
\label{sec:coq}

We have formalised the results in this paper using the Rocq proof system ({\ttfamily\href{https://rocq-prover.org}{rocq-prover.org}}).
The script is available at {\ttfamily\href{https://github.com/monstrencage/obs-alg-proofs}{github.com/monstrencage/obs-alg-proofs}}.
This was convenient to dispatch large and repetitive inductive proofs with some automation, and most importantly to guarantee that every case was properly considered.
This formalisation generally follows the same development as the paper.
Some variations were necessary, in particular because Rocq statements are made over types rather than sets.

An aspect worth discussing is our formalisation of graphs and cliques.
We represent a graph by a type $V$ for vertices, and a function $coh:V\to V\to \mathtt{Prop}$ for edges, plus the requirement that $coh$ is symmetric and reflexive.
As is usual in such developments, it was often convenient to also require the type $V$ to be decidable, i.e. there should be a boolean function $eq$ such that $u=v\Leftrightarrow eq\;u\;v=\top$.
A similar function to decide edges was also convenient.
More surprising is our encoding of cliques.
In our code, a clique may be seen as a pair $(m,i)$, where $m:V\to\set{\top,\bot}$ is a boolean membership predicate and $i:V\to V+\set{\bullet}$ is called the incompatibility function.
This pair is subject to the following requirements:
\begin{enumerate}
\item $m(a)= m(b)=\top\Rightarrow a\coh b$, i.e. any two vertices in the clique are connected;
\item $i(a)=b\in V\Rightarrow m(b)=\top\wedge b\incoh a$, i.e. when $i(a)$ is a vertex, it is a witness that $a$ cannot be added to the clique (it is incoherent with a member);
\item $i(a)=\bullet\Rightarrow \forall b,\,m(b)=\top\Rightarrow a\coh b$, i.e. when the incompatiblity function fails to yield a vertex, then $a$ may safely be added to the clique.
\end{enumerate}
This definition stems from a recurring need during the development to produce a witness that a given element is incompatible with a clique.
We implemented this for finite cliques, with a function that takes a list of vertices (with a proof that they are pairwise coherent), and produces a clique with the same members.

Other than this oddity, the rest of the development is a fairly straightforward adaptation of the ideas presented in this paper.

\subsection{On finite cliques}
\newcommand\implf{\mathbin{\rightarrow_f}}
\label{sec:fin-obs}

A natural variation on this algebra is its restriction to finite cliques.
This would naturally model a situation where the observer may only test a finite set of predicates on each event.

Formally, consider the set $\Cohf\eqdef\setcompr{\alpha\subseteq_f\atobs}{\forall a,b\in x,\,a\coh b}$ of finite cliques, i.e. $\Cohf=\Coh\cap\fpset{\atobs}$.
The ordering $\contains$ may be restricted to this set, thus yielding a closure operator $\closuref$ over $\pset{\Cohf}$.
Notice that it follows from the definitions that for any set $x$ of finite cliques, we have:
\begin{equation}
  x\closuref=x\closure\cap\Cohf.
  \label{eq:closf-clos}
\end{equation}

Like we did in the general case, we may use this to define the finitary observation algebra
$$\Obsf\eqdef\setcompr{x\closuref}{x\subseteq\Cohf},$$
and prove the following:
\begin{lemma}\label{lem:obsf-bdl}
  For any coherence graph $G$, $\tuple{\Obsf,\cup,\cap,\emptyset,\Cohf}$ is a bounded distributive lattice.
\end{lemma}
\begin{proof}
  Same as \Cref{prop:obs-bdl}.
\end{proof}

This lattice may be embedded as a sub-algebra of $\Obs$.
Before proving this, we establish the following identity:
\begin{lemma}\label{lem:closuref-closure}
   For any $x\subseteq\Cohf$, we have $x\closure= x\closuref\closure$.
\end{lemma}
\begin{proof}
  We proceed by mutual inclusion.
  \begin{itemize}
  \item Since $x\subseteq x\closure$ and $x\subseteq\Cohf$, we have $x\subseteq x\closure\cap\Cohf=x\closuref$.
    By monotonicity of Kuratowski closure operators~\cite[Theorem 1]{kuratowski1922}, we obtain $x\closure\subseteq x\closuref\closure$.
  \item Since $x\closuref=x\closure\cap{\Cohf}\subseteq x\closure$, by monotonicity and idempotence we get $x\closuref\closure\subseteq x\closure\closure=x\closure$.
    \qedhere
  \end{itemize}
\end{proof}

\begin{lemma}
  $\closure$ is an injective morphism of bounded distributive lattices from $\Obsf$ to $\Obs$.
\end{lemma}
\begin{proof}
  This proof is completely routine.
  \begin{itemize}
  \item $\emptyset\closure=\emptyset$ holds trivially;
  \item To check that $\Cohf\closure=\Coh$, it is enough to verify
    $\Coh\subseteq\Cohf\closure$.
    This is trivial because $\emptyset$ is a finite clique, hence $\Coh=\set{\emptyset}\closure\subseteq\Cohf\closure$.
  \item Since $\closure$ is a Kuratowski closure, it axiomatically commutes with unions.
  \item Let $x,y$ be two sets of finite cliques, and consider $\paren{x\closuref\cap y\closure}\closure$.
    We get:
    \begin{align*}
      \paren{x\closuref\cap y\closuref}\closure
      &=\paren{x\closuref\closure}\cap \paren{y\closuref\closure}\closure\\
      &=\paren{x\closure\cap y\closure}\closure\\
      &=x\closure\cap y\closure\\
      &=\paren{x\closuref\closure}\cap \paren{y\closuref\closure}
    \end{align*}
  \item Now for injectivity, assume $x,y\in\Obsf$ such that $x\closure=y\closure$.
    Using (\ref{eq:closf-clos}) and the fact that $x,y$ are finitary observations, we get
    \begin{align*}
      &x=x\closuref=x\closure\cap\Cohf
      =y\closure\cap\Cohf=y\closuref=y.\qedhere
    \end{align*}
  \end{itemize}
\end{proof}

\begin{corollary}
  $\tuple{\Obsf,\cup,\cap,\emptyset,\Cohf}$ is isomorphic to the sub-algebra of $\tuple{\Obs,\cup,\cap,\emptyset,\Coh}$ containing only those observations $x$ such that $\forall\alpha\in x,\,\exists\beta\in\Cohf:\,\alpha\contains\beta$.
\end{corollary}

However, this strong relation between the two algebras becomes undone when one considers the implication.
The finitary observation algebra does have implications, defined by setting for any $x,y\in\Obsf$:
$$x\implf y\eqdef\setcompr{\alpha\in\Cohf}{\forall\beta\in x,\,\alpha\cup\beta\in\Cohf\Rightarrow\alpha\cup\beta\in y}.$$
\begin{lemma}
  For any $x,y,z\in\Obsf$, we have $x\cap z\subseteq y$ if and only if $z\subseteq x\implf y$.
\end{lemma}
\begin{proof}
  We prove both implications.

  First, assume $x\wedge z\subseteq y$, and take $\alpha\in z$.
  To prove that $\alpha\in x\implf y$, we consider some arbitrary $\beta\in x$ such that $\alpha\cup\beta\in\Cohf$.
  Since both $x$ and $z$ are finitary observations, there are downclosed.
  We thus know that $\alpha\cup\beta\in x$ (since $\alpha\cup\beta\contains \beta\in x$) and $\alpha\cup\beta\in z$ (since $\alpha\cup\beta\contains \alpha\in z$).
  As such, we get $\alpha\cup\beta\in x\cap z$, which by hypothesis means $\alpha\cup\beta\in y$.
  This ensures that $\alpha\in x\implf y$.

  Second, we assert $z\subseteq x\implf y$ and pick $\alpha\in x\cap z$.
  By definition, $\alpha \in x$ and $\alpha\cup\alpha=\alpha\in\Cohf$.
  But we also have $\alpha\in z$, so by hypothesis $\alpha\in x\implf y$.
  Therefore, we deduce that $\alpha=\alpha\cup\alpha\in y$.
\end{proof}

However, the injective homomorphism we were able to build for the bounded distributive lattice structure does not preserve implications, or indeed pseudocomplements.
This seems unavoidable, as witnessed by the following example.
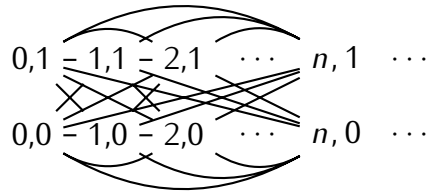
\begin{figure}[!ht]
  \centering
  \begin{tikzpicture}
    \node(0-1) at (2,1) {0,1};
    \node(1-1) at (3,1) {1,1};
    \node(2-1) at (4,1) {2,1};
    \node(n-1) at (6,1) {$n,1$};
    \node at ($(2-1)!.5!(n-1)$){$\cdots$};
    \node at ($(8,1)!.5!(n-1)$){$\cdots$};
    \node(0-2) at (2,0) {0,0};
    \node(1-2) at (3,0) {1,0};
    \node(2-2) at (4,0) {2,0};
    \node(n-2) at (6,0) {$n,0$};
    \node at ($(2-2)!.5!(n-2)$){$\cdots$};
    \node at ($(8,0)!.5!(n-2)$){$\cdots$};
    \draw[thick]
    (0-1)--(1-1)
    (1-1)--(2-1)
    (0-1) to[bend left](2-1)
    (0-1) to[bend left](n-1)
    (1-1) to[bend left](n-1)
    (2-1) to[bend left](n-1)
    (0-2)--(1-2)
    (1-2)--(2-2)
    (0-2) to[bend right](2-2)
    (0-2) to[bend right](n-2)
    (1-2) to[bend right](n-2)
    (2-2) to[bend right](n-2)
    (0-1) -- (1-2)
    (0-1) -- (2-2)
    (0-1) -- (n-2)
    (1-1) -- (0-2)
    (1-1) -- (2-2)
    (1-1) -- (n-2)
    (2-1) -- (1-2)
    (2-1) -- (0-2)
    (2-1) -- (n-2)
    (n-1) -- (1-2)
    (n-1) -- (2-2)
    (n-1) -- (0-2);
  \end{tikzpicture}
  \caption{The graph of $G_\omega$}
  \label{fig:g-omega}
\end{figure}
\begin{example}
  Consider the graph $G_\omega=\tuple{\Nat\times\set{0,1},\setcompr{\tuple{(n,i),(m,j)}}{i= j\vee n\neq m}}$, represented in~\Cref{fig:g-omega}.
  Its incoherence relation contains exactly the pairs $\tuple{(n,0),(n,1)}$ (and symetric).
  We define the set $x\eqdef\setcompr{\set{(n,0)}}{n\in\Nat}$ of singleton cliques whose second component is $0$.
  A clique $\alpha$ in $x\implf\bot$ would need to be incoherent with each $(n,0)$.
  The only way to accomplish this is to contain $(n,1)$, the only vertex incoherent with $(n,0)$.
  However, this entails that $\alpha$ is infinite, since it must contain the clique $\beta\eqdef\setcompr{(n,1)}{n\in\Nat}$.
  Therefore we may conclude that $x\implf\bot=\emptyset$, while $x\impl\bot=\set\beta\closure$.
\end{example}

\subsection{On the empty clique}
\label{sec:empty-clique}
\def\ne{\Coh^+}
The empty clique plays a special role in observation algebra: being the ``largest'' clique (in the sense of $\contains$), it is also the single generator of the interpretation of the constant $\top$.
However, in any observation algebra whose underlying graph is non-empty, there is another candidate for interpreting $\top$: the set of all non-empty cliques, which we will call $\ne$ in the following.
Indeed, $\ne$ sits just beneath $\set\emptyset\closure$, and above every other closed set of cliques.

The set $\ne$ is representable by a $\Tlat$ term for finite graphs, and by a $\Tobs$ term for infinite anticliques (see \Cref{sec:infinite-anticliques}).
This is not the case for infinite FAN graphs, since $\ne$ cannot be obtained as the closure of a finite set of finite cliques.
Similarly, $\ne$ cannot be represented by a $\Tobs$ term for infinite products, since it cannot be obtained from a finite set of pure observation terms (terms that only feature atomic observations from a single component).

This means that if one were to interpret $\top$ as $\ne$, for FAN graphs and for products, one would need special axioms for the finite case.
For instance, one could use $\top\equiv\bigvee_{a\in V}a$ to enumerate the vertices in the graph.
We have checked in Rocq that using this, axioms can be adapted to handle FAN graphs and infinite anticliques with this interpretation.

Products might be trickier to handle.
Indeed, one fact that underpins the proofs in \Cref{sec:prods} is that the empty clique is at the same time the largest clique in $\P=\bigotimes_{i\in\I}\G_i$ and in each of the $\G_i$.
In other words, the interpretation of $\top$ is somewhat uniform between the product algebra and the component algebras.
However, switching to $\ne$ as the interpretation of $\top$, we loose this property.
Instead, we have that $\ne(\G_i)\closure(\P)\subsetneq\ne(\P)$, because for any $j\neq i$ and $a\in O_j$ we know $\set{a\at j}\in \ne(\P)\setminus\paren{\ne(\G_i)\closure(\P)}$.
This property makes things much more complicated technically, and would require a different treatment altogether.

On the other hand, one natural graph construction might become more accessible.
We have not mentionned elsewhere in this paper the possibilty of combining graphs by sum rather than product, simply taking a disjoint union of vertices and edges without adding any edge between components.
Indeed, this rather useful construction has resisted our efforts to find a complete axiomatisation, because of the empty clique.
We believe that with the alternative $\ne$ interpretation of~$\top$, this construction would become similar to the case of products with~$\set\emptyset\closure$ as~$\top$.
We leave further study of these possibilities for future work.

\clearpage
\phantomsection{}
\addcontentsline{toc}{section}{References}
\bibliographystyle{apalike}
\bibliography{long,bibli}

\end{document}